\documentclass{article}
\usepackage[utf8]{inputenc}
\usepackage[english]{babel}

\usepackage{geometry}
\usepackage{amsmath,amsthm,amsfonts,amssymb,bm}
\usepackage{mathtools}
\usepackage{hyperref} 
\usepackage{url}
\usepackage{breakcites}
\usepackage{bbm} 
\usepackage{graphicx,subfigure,caption, float}
\usepackage{array}
\usepackage{color} 
\usepackage{colortbl}
\usepackage{newlfont}
\usepackage{mathrsfs}
\usepackage{lineno} 
\usepackage[toc,page]{appendix}
\usepackage{tikz}
\usepackage{pgfplots}
\usepackage{physics}
\usepackage{fancyhdr}
\usepackage{enumitem}
\usepackage{multirow}

\usepackage{crossreftools} 
\usepackage{apptools}

\addto\extrasenglish{} %
\addto\extrasenglish{} 

\usepackage{aliascnt}
\pgfplotsset{compat=1.16}   

\theoremstyle{plain}
\newtheorem{thrm}{Theorem}
    \newtheorem{theorem}{Theorem}[section]
    \newaliascnt{corollary}{theorem}
    
    \newtheorem{corollary}[corollary]{Corollary}
    \aliascntresetthe{corollary}

    \newaliascnt{lemma}{theorem}
    \newtheorem{lemma}[lemma]{Lemma}
    \aliascntresetthe{lemma}

    \newaliascnt{proposition}{theorem}
    \newtheorem{proposition}[proposition]{Proposition}
    \aliascntresetthe{proposition}

\theoremstyle{definition}
    \newaliascnt{definition}{theorem}
    \newtheorem{definition}[definition]{Definition}
    \aliascntresetthe{definition}

\theoremstyle{remark}
    \newaliascnt{remark}{theorem}
    \newtheorem{remark}[remark]{Remark}
    \aliascntresetthe{remark}

 \newaliascnt{example}{theorem}
    
    \aliascntresetthe{example}

\newcommand{\cP}{\mathcal{P}}
\newcommand{\cH}{\mathcal{H}}
\newcommand{\cV}{\mathcal{V}}
\newcommand{\cD}{\mathcal{D}}
\newcommand{\cO}{\mathcal{O}}

\newcommand{\cM}{\mathcal{M}}

\newcommand{\tq}{\tilde{\mathbf{q}}}
\newcommand{\tpsi}{\Tilde{\bm{\psi}}}
\newcommand{\td}{\tilde{d}}
\newcommand{\tY}{\tilde{Y}}
\newcommand{\tu}{\tilde{u}}
\newcommand{\tA}{\tilde{A}}

\newcommand{\tX}{\tilde{X}}
\newcommand{\tcO}{\tilde{\cO}}

\newcommand{\q}{\mathbf{q}}

\newcommand{\x}{\mathbf{x}}

\newcommand{\bfphi}{\bm{\phi}}
\newcommand{\bfpsi}{\bm{\psi}}

\newcommand{\bfxi}{\bm{\xi}}

\newcommand{\R}{\mathbb{R}}
\newcommand{\N}{\mathbb{N}}
\newcommand{\E}{\mathbb{E}\,}
\newcommand{\bP}{\mathbb{P}}
\newcommand{\Hmo}{\mathbf{H}^{-1}} 

\newcommand{\bLtwo}{\mathbf{L}^2}

\newcommand{\bH}{\mathbf{H}}

\newcommand{\Cd}{\mathcal{C}_{\tilde{d}}}
\newcommand{\CU}{\mathcal{C}_U}

\newcommand{\fa}{\quad \text{for all }\,}
\newcommand{\quand}{\quad \text{and} \quad }

\newcommand{\Ieps}{I_{\varepsilon}}
\newcommand{\scp}[1]{{\left\langle #1 \right\rangle}}

\DeclareMathOperator{\range}{range}
\DeclareMathOperator{\Law}{Law}

\newcommand{\ddt}{\tfrac{\text{d}}{\text{d}t}}

\newcommand{\vertiii}[1]{{\left\vert\kern-0.25ex\left\vert\kern-0.25ex\left\vert #1 
    \right\vert\kern-0.25ex\right\vert\kern-0.25ex\right\vert}}

\RequirePackage[normalem]{ulem} 
\RequirePackage{color}\definecolor{RED}{rgb}{1,0,0}\definecolor{BLUE}{rgb}{0,0,1} 

\title{Linear and fractional response for nonlinear dissipative SPDEs}
\author{%
  Giulia Carigi
  \and Tobias Kuna
  \and Jochen Br\"{o}cker
  }
\date{21/10/2022}
\begin{document}
\maketitle
\begin{abstract}
    A framework to establish response theory for a class of nonlinear stochastic partial differential equations (SPDEs) is provided. More specifically, it is shown that for a certain class of observables, the averages of those observables against the stationary measure of the SPDE are differentiable (linear response) or, under weaker conditions, locally H\"{o}lder continuous (fractional response) as functions of a deterministic additive forcing. 
    The method allows to consider observables that are not necessarily differentiable.
    For such observables, spectral gap results for the Markov semigroup associated with the SPDE have recently been established that are fairly accessible.
    This is important here as spectral gaps are a major ingredient for establishing linear response.
    The results are applied to the 2D stochastic Navier--Stokes equation and the stochastic two--layer quasi--geostrophic model, an intermediate complexity model popular in the geosciences to study atmosphere and ocean dynamics. 
    The physical motivation for studying the response to perturbations in the forcings for models in geophysical fluid dynamics comes from climate change and relate to the question as to whether statistical properties of the dynamics derived under current conditions will be valid under different forcing scenarios.
\end{abstract}


\section{Introduction}
In this work we consider a framework suitable to establish response theory for a class of nonlinear stochastic partial differential equations including the 2D stochastic Navier--Stokes equation as well as the stochastic two--layer quasi--geostrophic model, an intermediate complexity model popular in the geosciences to study atmosphere and ocean dynamics. 
By studying response (linear and fractional) we provide an insight into how the long term statistical properties of the model of interest are affected by small changes in the parameters of the system, namely whether the statistics of observables under the current set of parameters will change little under small perturbations of the parameters, and derive a formula for the change of the statistics. In particular, by studying the response to perturbations in the parameters for models in geophysical fluid dynamics, like the two--layer quasi--geostrophic model, we give a mathematical interpretation of the question whether statistical properties derived under current conditions will be valid under future climates. For more on the relevance of linear response theory in geophysics see for example applications like \cite{AbramovMadja2012,Majda581, lucarini2017predicting} or the recent review paper \cite{GhilLucarini20}.

More specifically, consider a family of dynamical systems depending on a parameter and admitting an invariant measure. By {\em linear response} we mean the differentiability of the family of invariant measures with respect to the parameter, and by {\em fractional response} we mean H\"{o}lder continuity of the invariant measures in the parameter.
In fact, even though the invariant measures are often a very singular object, they can nonetheless change smoothly with respect to changes in the parameters, at least in a weak sense. In case of linear response in particular one aims at a \emph{response formula} that is and expression for the derivative of the invariant measure exclusively in terms of objects related to the unperturbed dynamics. In the applications this would mean that one can infer properties of the perturbed dynamics from those of the unperturbed.

Before considering response theory for dissipative SPDEs, let us first describe the known results for finite dimensional systems, where there exists a large body of mathematical literature on linear response. 
For hyperbolic systems, in absence of stochasticity, the pioneering work of Ruelle \cite{ruelle1997differentiation} ensured the differentiability of invariant measures, in particular of SRB measures which carry a certain physical interpretation. The result has been extended also to partially hyperbolic systems in \cite{dolgopyat2004differentiability} but little is known for other classes of deterministic systems, finite or infinite dimensional. In particular the existence of SRB measures for Navier-Stokes is entirely open. Equations of fluid dynamic seems out of scope to be treated with techniques used for finite dimensional dynamics.
For a review on linear response theory in deterministic systems see the survey article \cite{Baladi_Else}. 
For stochastic systems, the impact of stochastic perturbations on Ruelle's linear response has been investigated for example in \cite{lucarini2012stochastic}. Recent works \cite{Galatolo_2019} and \cite{BAHSOUN2020107011} pioneered linear response in finite dimensional random dynamical systems.
However less is known for infinite dimensional systems associated to stochastic partial differential equations.

To the best of the authors' knowledge, the only result which covers dynamical systems associated to a large class of stochastic partial differential equations is the work of Hairer and Madja \cite{HMajda10}. 
%
    The authors proved the weak differentiability of the unique invariant measures $\mu_a$ of families of Markov semigroups $\lbrace \cP_t^a \, : \, t\geq 0, a \in \R\rbrace$ on a Hilbert space $\cH$. This means that, given any $a_0\in \R$, the map 
    \begin{equation*}
        a \mapsto \langle \varphi, \mu_a \rangle = \int_\cH \varphi(x) \, \mu_a(dx)
    \end{equation*}
    is differentiable at $a_0$ for an appropriate class of test functions $\varphi$. Moreover an explicit expression for its derivative is provided. 
%
    In this paper we will reformulate the framework of \cite{HMajda10} and give a set of sufficient conditions for linear response in a general space of observables (\autoref{subsec:linear}), before providing a verifiable set of conditions for a class of SPDEs (\autoref{subsec:SPDElinear}) using a different space of test functions than in \cite{HMajda10}.

One crucial (but not necessary) condition to establish linear response for a given value $a_0$ of the parameter $a$ is that for a $t>0$ the operator $\cP_t^{a_0}$ has a spectral gap on an appropriate class of observables. This means that $\cP_t^{a_0}$ has one as simple eigenvalue and the remaining spectrum is concentrated in a disk of radius strictly smaller than one. 
As a consequence, the operator $\cP^{a_0}_{t}$ will have a resolvent on the space of observables modulo the constant functions (which constitute the eigenspace associated to the eigenvalue one).
In \cite{HMajda10} the authors consider as observables the closure of the space of smooth observables with respect to a weighted $C^1$--norm. 
Indeed in \cite{HMatt08}, the spectral gap property for such observables was proven for the 2D Navier-Stokes equations with highly degenerate noise. 
%
The framework developed here builds on a different approach to show the spectral gap developed in \cite{HMattSch11}. This result provides sufficient conditions for the spectral gap property to hold in a space of H\"{o}lder--type functions, which are therefore less regular than those used in \cite{HMajda10} and \cite{HMatt08}. This methodology is simpler to verify than the approach from \cite{HMatt08, HMajda10} as it does not require the use of Malliavin calculus for example. Furthermore, this methodology allows us to give quite concrete conditions for a wide class of SPDEs, since recently \cite{butkovsky2020} provided sufficient conditions, particularly suitable for SPDEs, for the results in \cite{HMattSch11} to apply. Since we use a different spectral gap result than \cite{HMajda10} however, we cannot deal with highly degenerate noises and we need to impose stronger conditions on the nature of the perturbations to the dynamics we study the response for.  
We focus on dissipative nonlinear and stochastic equations with an additional deterministic forcing, and study the dependence of the invariant measure of this equation on the forcing strength. In particular we obtain differentiability and a linear response formula for forcings which are in the range of the covariance of the noise (in a sense to be made precise). 
%

%
For forcings not satisfying such conditions we can nevertheless show weak H\"{o}lder continuity of the invariant measure, also referred to as fractional response (see e.g. \cite{baladi2017linear}).
This result does not provide a linear approximation of the perturbed dynamics in terms of the unperturbed one, but still ensures that a small change in the intensity of the forcing does not cause a discontinuity in the long time average behaviour of observables which are themselves at least H\"{o}lder continuous. Several results in the literature establish continuity properties in weaker topologies of the invariant measure on model parameters. For example in \cite[Section~5.5]{HMatt08} it is shown that the invariant measure of 2D Navier--Stokes equation with additive noise is locally continuous in a Wasserstein distance with respect to the model parameters. 
\subsection{Overview over results}
In Section~\ref{sec:response_method} we establish an abstract response result in the following setup.
Let $\Ieps := (a_0 - \varepsilon, a_0 + \varepsilon) \subset \R$ be an interval and $\{\cP^a \, : \, a \in \Ieps \rbrace $ a family of Markov transition kernels acting on a Hilbert space~$\cH$, with $\mu_a$ an invariant measure of $\cP^a$ for each $a \in \Ieps$.
Let $\cO$ be a Banach space of measurable functions (``observables'') on $\cH$ so that $(1 - \cP^{a}) \psi + c \in  \cO$ whenever $\psi \in \cO$ and $c \in \R$.
Finally, for $\|\cdot\|_U$ (the supremum--norm with weight function $U$) define the space $\CU := \overline{\cO}$, where the closure is with respect to the norm $\|\cdot\|_U$ (this norm is assumed to be finite on $\cO$).
Suppose the following conditions hold:
\begin{description}
        \item[(Spectral gap)] There exists $\rho<1$ such that 
    \begin{equation*}
        \|\cP^{a_0} \varphi - \scp{\varphi, \mu_{a_0}} \|_{\cO} \leq \rho \| \varphi - \scp{\varphi, \mu_{a_0}}\|_{\cO} \fa \varphi \in \cO;
    \end{equation*}
  \item[(Differentiability)]
for every $\psi \in \cO$ the map $a \mapsto (1 - \cP^a)\psi$ has values in $\CU$ and is differentiable at $a = a_0$, and furthermore
            \begin{equation*}
                \|D_a \cP^a \psi |_{a = a_0} \|_U \leq C \|\psi \|_{\cO};
            \end{equation*}
  \item[(Uniform integrability)] $\sup_{a \in \Ieps} \|\mu_a\|_U \leq C_{\varepsilon} $.
              \end{description}
 Then the following version of \cite[Theorem~2.3]{HMajda10} holds:
 \renewcommand{\thethrm}{\ref{thm:general_response_HM}}
  \begin{thrm}[Page~\pageref{thm:general_response_HM}]
        Under the stated assumptions, the map $a \mapsto \scp{ \varphi, \mu_a} $ is differentiable at $a = a_0$ for every $\varphi\in \cO$ and in particular 
        \begin{equation*}
            \left. \dv{}{a} \scp{ \varphi, \mu_a }\right|_{a = a_0} = \scp{ D_a \cP^{a}|_{a = a_0} (1 - \cP^{a_0})^{-1}(\varphi - \scp{\varphi, \mu_{a_0}}), \mu_{a_0}}. 
        \end{equation*}
    \end{thrm}
H\"{o}lder continuity of the invariant measure with respect to the parameter is the subject of Theorem~\ref{thm:general_holder}.
The conditions are very similar to those of Theorem~\ref{thm:general_response_HM}, except that the differentiability condition is replaced by the requirement that the mapping $a \to \cP^{a}$ is H\"{o}lder continuous as a mapping from $\Ieps$ to $L(\cO, \CU)$, the space of linear operators between $\cO$ and $\CU$, with the operator norm.
In Section~\ref{sec:response_SPDE} we specialise the methodology to Markov processes generated by dissipative SPDE's driven by moderately degenerate additive noise and study the response with respect to the amplitude of a deterministic external forcing. 
In Theorem~\ref{thm:expconvergence} the spectral gap for observables that are basically H\"{o}lder continuous functions obtained in \cite{carigi2022exponential} is recalled. 
Yet this brings about a difficulty when proving the differentiability of the Markov kernel, as the observables are not differentiable with respect to the forcing for individual realisations of the process.
This difficulty is circumvented by using a Girsanov transform to compare the distribution of the process for different forcings; for this to work however, the forcing has to be in the range of the noise covariance.
Under that condition, Theorem~\ref{thm:responseU} establishes linear response.
Without any condition on the range of the noise covariance, we can still establish fractional response in Theorem~\ref{thm:SPDEHolder} under assumptions that are otherwise the same as for linear response (except that the forcing needs to be H\"{o}lder continuous in the parameter, only).
In Sections~\ref{sec:response_NS} and~\ref{sec:response_QG}, we apply the developed methodology to the two--dimensional stochastic Navier--Stokes equation (Theorems~\ref{thm:NSresponse} and~\ref{thm:NSHolder} for linear and fractional response, respectively), and to the stochastic two--layer quasi--geostrophic (2LQG) model (Theorems~\ref{thm:QGresponse} and~\ref{thm:QGHolder} for linear and fractional response, respectively). 
 The 2LQG equations model mid--latitude atmosphere and ocean dynamics at large scale.
 The model describes two layers of fluid one on top of the another.
The fluid experiences the Coriolis effect and an external forcing which acts only on the top layer and has a non-trivial stochastic part, accounting for example for wind shear on the surface.
For a more detailed exposition of the mathematical description see for example \cite{carigi2022exponential} and references therein.

The model equations are
\begin{align}
\label{eq:QG_stochastic_intro}
\begin{split}
    &d q_1 + J(\psi_1, q_1 + \beta y ) \, dt = \left(\nu\Delta^2\psi_1 \,+ f(a) \right) dt + d W\\
    & \partial_t q_2 + J(\psi_2, q_2 + \beta y ) = \nu\Delta^2\psi_2 - r\Delta \psi_2,
\end{split}
\end{align}
on a squared domain $\cD = [0,L]\times [0,L]\subset \R^2$ where $\psi_1, \psi_2$ are the streamfunctions and $q_1, q_2$ are the quasi--geostrophic potential vorticities of the upper and lower layer, respectively; they are related through
 \begin{equation*}
\begin{split}
     q_1 = \Delta \psi_1 + F_1(\psi_2 - \psi_1) \\
     q_2 = \Delta \psi_2 + F_2(\psi_1 - \psi_2),
\end{split}
\end{equation*}
where $F_1, F_2$ are positive constants. 
Moreover, $J$ is the Jacobian operator $J(a,b) = \nabla^{\perp}a \cdot \nabla b$, $W$ is a Wiener process with covariance $Q$ which is a nonnegative, symmetric operator which we take to be trace class in $L^2$. We also assume that $Q$ and $-\Delta$ commute.
In \cite{carigi2022exponential} it is shown that \eqref{eq:QG_stoc_vec} exhibits a spectral gap whenever the parameter $r$ in Equation~\eqref{eq:QG_stochastic_intro} (which is related to be bottom friction) is large enough (in terms of other parameters of the model, see Sec.~\ref{subsec:QGspectralgap} for details).
In this situation, we have the following result:
     \renewcommand{\thethrm}{\ref{thm:QGresponse}}
 \begin{thrm}[Page~\pageref{thm:QGresponse}]
         Consider the two--layer quasi--geostrophic equation \eqref{eq:QG_stoc_vec} with $f(a)$ continuously differentiable as a function from $\R$ into $\range Q$ with $\left|Q^{-1/2}D_af(a)\right|$ locally uniformly bounded in $a$, and invariant measure $\mu_a$. If condition \eqref{eq:condition_r} on $r$ is met, the map $a \mapsto\langle \varphi, \mu_a\rangle $ is differentiable at $a = a_0$ for every $\varphi\in \Cd$ with 
         \begin{equation*}
            \left. \dv{}{a} \langle \varphi, \mu_a \rangle\right|_{a = a_0} = \langle D_a \cP_t^{a}|_{a = a_0} (1 - \cP_t^{a_0})^{-1}(\varphi - \langle \varphi, \mu_{a_0}\rangle  , \mu_{a_0}\rangle. 
        \end{equation*}
    \end{thrm}
    %
%
Theorem~\ref{thm:NSresponse} provides a very similar result for 2D Navier--Stokes with additive noise. 
For both equations, the core condition to establish linear response is that the force lies within the range of the noise covariance. If $f$ does not take values in the range of $Q$ then in \autoref{thm:QGHolder} and \autoref{thm:NSHolder} we show that $a \mapsto \langle \varphi, \mu_a \rangle$ is H\"{o}lder continuous, this holds when $a\mapsto f(a)$ is itself only H\"{o}lder continuous. 
Spectral gap results have been demonstrated for example in \cite{butkovsky2020} for Navier--Stokes and in \cite{carigi2022exponential} for the 2LQG model.
We expect the results presented here to be straightforwardly applicable to other model with similar structure.
\paragraph{Acknowledgments}
The work presented here greatly benefited from fruitful discussions with a number of colleagues. 
In particular, we are very grateful to Benedetta Ferrario, Franco Flandoli, Valerio Lucarini, and Jeroen Wouters for criticisms, comments, suggestions, and encouragement.
GC's work was funded by the Centre for Doctoral Training in Mathematics of Planet Earth, UK (EPSRC grant agreement EP/L016613/1), by an LMS Early Career Fellowship (grant ECF1920-48), and by a postdoctoral fellowship (EPSRC grant agreement EP/W522375/1).
    
\section{General Methodology}
\label{sec:response_method}
%
  \subsection{A perturbation identity}\label{subsec:response_setup}
    Consider an open interval $I \subset \R$ and for each $a \in I$, let $\{\cP^a_t\}_{t \geq 0}$ be a Markov semigroup on $\cH$ depending on the parameter $a$, with $\mu_a$ be a corresponding invariant probability measure.
        We want to study the regularity of the averages $\scp{\varphi, \mu_a}$ with respect to $a$ for suitable observables $\varphi$.
        In particular, we aim to establish conditions that ensure differentiability (linear response) and local H\"{o}lder continuity (fractional response) of $a \mapsto \scp{\varphi, \mu_a}$.
        Let us start by showing a simple yet crucial identity.
        
        Fix a time $t$ and for notational simplicity drop the dependence on time of the semigroup.
        We consider a space $\cO$ of measurable (but not necessarily bounded) functions $\psi$ on $\cH$ so that $(1 - \cP^{a_1})\psi$ is integrable with respect to $\mu_{a_2}$ for any $a_1, a_2 \in I$.
        Further, for each $a \in I$ we let $\tcO_a$ be the space of functions $\varphi$ which can be represented in the form $\varphi = (1 - \cP^a)\psi + c$ for some $\psi \in \cO$ and constant $c \in \R$.
        Since $\scp{(1 - \cP^{a}) \psi, \mu_a} = 0$ for any $a \in I$ by the invariance of $\mu_a$, we find $c = \scp{\varphi, \mu_a}$.
        The announced identity is 
        \begin{equation}
    \label{eq:formal_inversion_(1-P)}
        \scp{\varphi, \mu_{a_2} - \mu_{a_1}} = \scp{(\cP^{a_2}- \cP^{a_1}) \psi, \mu_{a_2}},
        \end{equation}
        for any $\varphi \in \tcO_{a_1}$, where $\psi \in \cO$ is a solution of $\varphi = (1 - \cP^{a_1}) \psi + c$.
To see this, we use the invariance several times to write the right hand side of Equation~\eqref{eq:formal_inversion_(1-P)} as
        \begin{equation*}
          \begin{split}
            \scp{ (\cP^{a_2} - \cP^{a_1}) \psi, \mu_{a_2} },
            & = \scp{ (1 - \cP^{a_1}) \psi, \mu_{a_2} }
            - \scp{ (1 - \cP^{a_2}) \psi, \mu_{a_2} } \\
        & = \scp{ (1 - \cP^{a_1}) \psi, \mu_{a_2} } \\
        & = \scp{ (1 - \cP^{a_1}) \psi, \mu_{a_2} }
        - \scp{ (1 - \cP^{a_1}) \psi, \mu_{a_1} }\\
        & = \scp{ \varphi , \mu_{a_2}  - \mu_{a_1} }.
        \end{split}
        \end{equation*}
Note that all terms in this calculation are well defined due to our integrability assumption.
Equation~\eqref{eq:formal_inversion_(1-P)} demonstrates that differentiability or local H\"{o}lder continuity of $a \mapsto\scp{ \varphi, \mu_a }$ hinges on the corresponding property of the semigroup which we will consider in the next sections.
Furthermore, the fact that $\tcO_a$ depends on $a$ is an issue, and the dependence of $\psi$ on $\varphi$ will need to be clarified.
    We will address this using the following result although alternatives are certainly possible.
    \begin{proposition}
      \label{prop:spectralgap_invertibility}
        Fix $a \in I$ and suppose that $(\cO, \| \cdot \|_{\cO} )$ is a Banach space. Further, assume that $\tcO_a \subset \cO$, that $\ker \mu_a$ is a closed subset of $\cO$, and that there exists $\rho<1$ and $t>0$ such that 
        \begin{equation}
        \label{eq:ch5spectralgapO}
            \|\cP^{a}_t \varphi - \scp{\varphi, \mu_a} \|_{\cO} \leq \rho \,\| \varphi - \scp{\varphi, \mu_a} \|_{\cO} \fa \varphi \in \cO.
        \end{equation}
        Then the following holds:
        \begin{enumerate}
          \item $1 - \cP^{a}_t$ is invertible on $\ker\mu_a$ and $\|\left.(1 - \cP^{a}_t)\right|_{\ker \mu_a} ^{-1}\|_{\cO} \leq \frac{1}{1 - \rho}$;
        \item Every $\varphi \in \cO$ can be represented in the form $\varphi = (1 - \cP^{a}_t) \psi + \scp{\varphi, \mu_a}$ with $\psi  = (1 - \cP^{a}_t)^{-1} (\varphi - \scp{\varphi, \mu_a}) \in \ker\mu_a$; in particular, we have $\tcO_a = \cO$;
      \item For $\psi$ from the previous item we have the bound $\|\psi\|_{\cO} \leq \frac{1}{1 - \rho} \| \varphi - \scp{\varphi, \mu_a} \|_{\cO}$.
        \end{enumerate}
    \end{proposition}
    %
    %
%
    It follows that if the conditions in Proposition~\ref{prop:spectralgap_invertibility} are met, then Equation~\eqref{eq:formal_inversion_(1-P)} holds for all $\varphi\in \cO$.
    \begin{remark}\label{remark_centred1}
      If $\tcO_a \subset \cO$ then $\cO$ contains the constant functions.
      Furthermore $1 - \cP^{a}$ and hence also $\cP^{a}$ are operators on $\cO$.
      This means that both sides of Equation~\eqref{eq:ch5spectralgapO} are well defined.
      Furthermore, that equation implies the existence of a {\em spectral gap} between the eigenvalue~1 of $\cP^a$ (with eigenspace being the constant functions) and the remainder of the spectrum of $\cP^a$ on $\cO$ being confined to a circle of radius $\rho < 1$.
      Proposition \ref{prop:spectralgap_invertibility} then states, broadly speaking, that $(1 - \cP_t^{a})$ is invertible on $\ker\mu_a$.
      \end{remark}
\subsection{Linear Response}\label{subsec:linear}
We now focus on the differentiability in $a$ of $\scp{ \varphi, \mu_{a} }$.
In order to simplify the notation, we assume, without loss of generality, that $a = a_0$ and that the interval $I$ is of the form $\Ieps := (a_0-\varepsilon, a_0+ \varepsilon)$, for $\varepsilon>0$; we will also write $\tcO$ instead of $\tcO_{a_0}$.
     In \autoref{thm:general_response} we provide general conditions on $\cO$ and the semigroup which ensure that  $\scp{ \varphi, \mu_{a} }$ with $\varphi\in \tcO$ is differentiable in $a = a_0$. 
     If in addition the conditions of \autoref{prop:spectralgap_invertibility} hold at $a = a_0$, it will be demonstrated that $\scp{ \varphi, \mu_{a }}$ with $\varphi\in \cO$ is differentiable in $a = a_0$ (see \autoref{thm:general_response_HM}).
      Given a suitable space $\cO$, it follows from Equation~\eqref{eq:formal_inversion_(1-P)} that in order to compute the derivative of $\scp{ \varphi, \mu_a }$ at $a = a_0$ for some $\varphi\in \tcO$, we need to show that the following limit exists:
      \begin{equation}
        \label{equ:limit_zero}
        \lim_{a \to a_0}\scp{  \frac{(\cP^a- \cP^{a_0})
    \psi}{a - a_0}, \mu_a } , 
    \end{equation}
 where $\psi \in \cO$ such that $\varphi  = (1 - \cP_t^{a_0})\psi + c$. 
    Broadly speaking, if $\cP^a \psi$ is differentiable in $a = a_0$ and $\mu_a$ is weakly continuous in $a = a_0$, with some uniformity in $a$, we can expect the desired convergence to hold. 
    To formalise this idea, it turns out to be convenient in applications to introduce a further Banach space $\CU$ of measurable functions on $\cH$ with norm
    \begin{equation}
    \label{eq: def norm U}
        \|\varphi\|_{U} := \sup_{x\in \cH}  \frac{|\varphi(x)|}{U(x)},
    \end{equation} 
    where $U: \cH\to [1, \infty)$ is assumed integrable with respect to $\mu_a$ for any $a \in \Ieps$.
      (This condition is preliminary only and will need to be strengthened later, see \eqref{eq: condition U} below.)
      For every $\psi \in \cO$ consider the map $a \mapsto (1 - \cP^a) \psi$.
      Suppose it is differentiable at $a = a_0$ as a function from $\Ieps$ into $\CU$, that is, there exists an element of $\CU$ which we denote by $D_a (\cP^a \psi)|_{a = a_0}$ such that 
    \begin{equation}
    \label{eq:assump2HMadja2}
        \dfrac{ (\cP^a- \cP^{a_0})\psi}{a - a_0} =  \left. D_a (\cP^a \psi)\right|_{a = a_0} + r_a,
    \end{equation}
   where the remainder $r_a$ is a function from $\Ieps$ to $\CU$ with $\|r_a\|_{U} \to 0 $ for $a \to a_0$.
   If in fact $\tcO \subset \CU$ (as it will be the case in our applications), the differentiability of $\cP^{a}$ as a function $I_\varepsilon \rightarrow \CU$ is a weaker requirement than the differentiability as a function $I_\varepsilon \rightarrow \tcO$.
   In our applications to SPDE's, $U$ will be an unbounded function and related to how the a priori estimates depend on the initial condition.
   It might seem more natural to define $D_a (\cP^a \psi) |_{a = a_0}$ as the derivative of $ \cP^a \psi$ at $a = a_0$, rather than of $(1 - \cP^a) \psi$ as we do it here; the reason for doing it this way though is that $1 - \cP^a$ is required  to have values in $\CU$ which is not a~priori true for $\cP^a$.
   Using Equation~\eqref{eq:assump2HMadja2} in~\eqref{equ:limit_zero}, we conclude that the existence of 
    \begin{equation}
    \label{eq:assump2HMadja}
    \lim_{a \to a_0} \left( \scp{ \left. D_a (\cP^a \psi)\right|_{a = a_0}, \mu_a}
    + \scp{r_a, \mu_a} \right)
    \end{equation}
needs to be established.
   Regarding the term involving the remainder $r_a$, note that
    \begin{equation}
    \label{eq: < R_a, mu_a >}
        |\scp{ r_a, \mu_a }| \leq \scp{ \left|\tfrac{r_a}{U}U\right| , \mu_a } \leq \| r_a\|_U \scp{ U ,  \mu_a}.
    \end{equation}
     Therefore, if there is an $\varepsilon$ such that 
    \begin{equation}
    \label{eq: condition U}
        \sup_{a \in \Ieps} \scp{ U, \mu_a } < \infty,
    \end{equation}
    then $ |\scp{ r_a, \mu_a }|$ vanishes in the limit $a \to a_0$ as desired for all $a \in \Ieps$.
    (Condition~\eqref{eq: condition U} is the stronger integrability condition alluded to earlier.)
   In particular, Condition~\eqref{eq: condition U} ensures that all $\xi \in \CU$ are integrable with respect to $\mu_a$ for any $a \in I_\varepsilon$. 
   For the first term of \eqref{eq:assump2HMadja} we have to show that
   \begin{equation*}
       \scp{ \left. D_a (\cP^a \psi)\right|_{a = a_0} , \mu_a - \mu_{a_0}} \rightarrow 0 .
   \end{equation*}
   Hence if we assume that for all $\xi \in \CU$ we have 
    \begin{equation}
    \label{eq:convergenceCU}
    \scp{  \xi , \mu_a } \to \scp{ \xi , \mu_{a_0} }
    \end{equation}%
    for $a \to a_0$, then, as $\left. D_a (\cP^a \psi )\right|_{a = a_0}\in \CU$,
    \[ 
       \dfrac{\scp{ (\cP^a- \cP^{a_0})\psi, \mu_a }}{a -a_0} \to \scp{ \left. D_a (\cP^a \psi )\right|_{a = a_0} , \mu_{a_0}}
    \] 
    as desired. 
   In summary, if $\varphi \in \tcO$ (we recall that this means there exists $\psi\in \cO$ and $c \in \R$ such that $\varphi = (1 - \cP^{a_0})\psi + c$), 
   then, thanks to Equation~\eqref{eq:formal_inversion_(1-P)} and the differentiability condition\eqref{eq:assump2HMadja2} as well as the assumptions~(\ref{eq: condition U},\ref{eq:convergenceCU}), we get for $a \to a_0$
    \begin{equation*}
         \left| \dfrac{\scp{ \varphi, (\mu_a - \mu_{a_0})}}{a - a_0} - \scp{ \left. D_a (\cP^a \psi)\right|_{a = a_0}, \mu_{a_0}} \right|
         \rightarrow 0.
    \end{equation*}
We have shown the following result, a generalisation of~\cite[Theorem~2.3]{HMajda10}. 
    \begin{theorem}
    \label{thm:general_response}
    Let $\Ieps := (a_0 - \varepsilon, a_0 + \varepsilon) \subset \R$ be an interval and $\lbrace \{\cP^a_t\}_{t \geq 0} \, : \, a \in \Ieps \rbrace $ a family of Markov semigroups acting on a Hilbert space~$\cH$, with $\mu_a$ an invariant measure of $\{\cP^a_t\}_{t \geq 0}$ for each $a \in \Ieps$.
      Further, let $\cO$ be a Banach space of measurable functions on $\cH$.
    Suppose there exists a function $U: \cH \to [1, \infty)$ and some $t > 0$ such that the following conditions hold:
        \begin{enumerate}[label=(\roman*)]
            \item\label{item:semigroup differentiable} for every $\psi \in \cO$ the map $a \to (1 - \cP_t^a) \psi$ has values in  $\CU$ and is differentiable at $a = a_0$;
            \item\label{item:weak convergence} For any $\xi\in \CU$ we have $\scp{  \xi , \mu_a} \to \scp{  \xi, \mu_{a_0}}$ for $a \to a_0$; 
            \item\label{item:<U, mu>} $\sup_{a \in \Ieps} \scp{ U, \mu_a} < \infty$. 
        \end{enumerate}
        Then the map $a \mapsto \scp{ \varphi, \mu_a} $ is differentiable at $a = a_0$ for every $\varphi$ of the form $\varphi = (1 - \cP_t^{a_0})\psi + c$ with some $\psi \in \cO$, and we have the {\em response formula}
        \begin{equation*}
            \left. \dv{}{a} \scp{ \varphi, \mu_a }\right|_{a = a_0} =  \scp{ \left. D_a (\cP_t^a \psi)\right|_{a = a_0} , \mu_{a_0}}. 
        \end{equation*}
    \end{theorem}
    Note that conditions~\ref{item:semigroup differentiable}~\ref{item:<U, mu>} together ensure the integrability of elements in $\tcO$ with respect to $\mu_a$ for any $a \in I_\varepsilon$.
Furthermore, condition~\ref{item:<U, mu>} implies that it is sufficient to establish condition~\ref{item:weak convergence} for $\xi$ in a dense set of $\CU$ in order to obtain it for all of $\CU$.
      %
   
    On the other hand, one might hope that, in view of the identity~\eqref{eq:formal_inversion_(1-P)}, conditions like (iii) and (i) may imply the weak convergence in condition (ii). Indeed, uniform bound on the derivative implies that the r.h.s of \eqref{eq:formal_inversion_(1-P)} is even Lipschitz continuous in the parameter $a$ for $\xi\in \tcO$. Hence, we need a condition that $\tcO$ is large enough such that one can recover convergence for $\xi \in \CU$. This is the role of the condition 2 in the following proposition.
    \begin{proposition}
    \label{prop:weak_convergence_CU}
    Consider the same setup as in Theorem~\ref{thm:general_response} and assume that condition~\ref{item:<U, mu>} of that theorem holds, but we replace conditions~(i,ii) with the following: 
    \begin{enumerate}
    \item for every $\psi \in \cO$  the map $a \mapsto (1 - \cP_t^a)\psi$ has values in $\CU$ and is differentiable at $a = a_0$, and furthermore
            \begin{equation}
            \label{eq:ch5_bound_derivative}
                \|D_a \cP_t^a \psi |_{a = a_0} \|_U \leq C \|\psi \|_{\cO}. 
            \end{equation}
          \item\label{itm:weak_convergenceCU2} 
          $\overline{\Im (1 - \cP_t^{a_0})} = \ker \mu_{a_0}$ where $1 - \cP_t^{a_0}$ is understood as a map from $\cO$ to $\CU$ and both the closure as well as the kernel are to be understood in $\CU$.
        \end{enumerate}
    Then for any $\xi \in \CU$, $\scp{ \xi, \mu_a } \to \scp{ \xi, \mu_{a_0} }$ when $a \to a_0$.
    \end{proposition}
    %
        %
    \begin{proof}
      To start with, it is clear that the result is linear in $\xi$ and correct whenever $\xi$ is constant function, so it is sufficient to prove it for $\xi \in \ker(\mu_{a_0})$ as every $\xi \in \CU$ can be written as $\xi = \xi' + c$ where $\xi' \in \ker(\mu_{a_0})$.
      Furthermore, using the shorthand $C_\varepsilon := \sup_{a \in \Ieps} |\scp{ U , \mu_a }|$, which is finite by condition~\ref{item:<U, mu>}, it is easy to see that
              \begin{equation}
        \label{eq:ch5<xi,mu>}
        \left| \scp{ \xi, \mu_a } \right| \leq C_\varepsilon \|\xi\|_U,
                \end{equation}
             that is, the probability measures $\{\mu_a, a \in I\}$ are functionals on $\CU$ with uniformly bounded operator norm (the calculation is the same as in Eq.~\ref{eq: < R_a, mu_a >}).
             This implies that it is sufficient to prove the result for all $\xi$ in a dense set of $\ker(\mu_{a_0})$, for which we can take $\Im{1 - \cP^{a_0}_t}$ by condition~\ref{itm:weak_convergenceCU2}.
             Hence, we may assume $\xi = (1 - \cP^{a_0})\psi$ for some $\psi \in \cO$.
             %
             Next we use \eqref{eq:formal_inversion_(1-P)} and the differentiability to obtain
         \begin{equation*}
               \left| \scp{ \xi, \mu_a - \mu_{a_0} } \right| 
\leq |a - a_0|\,  \left| \scp{ D(\cP_t^{a}\psi)|_{a = a_0}, \mu_a } + \scp{ r_a, \mu_a } \right|.
         \end{equation*}
                 Thanks to Equations~(\ref{eq: < R_a, mu_a >},\ref{eq:ch5_bound_derivative},\ref{eq:ch5<xi,mu>}), the right hand side is bounded as follows
        \[ 
        \leq |a - a_0| C_\varepsilon (C \|\psi\|_{\cO} + \|r_a\|_U)
        \]
        for all $a\in \Ieps$.
        This demonstrates that $\scp{ \xi, \mu_a } \to \scp{ \xi, \mu_{a_0} }$ for $a \to a_0$, as desired, and even that this function is locally Lipschitz with respect to the parameter~$a$.
    \end{proof}

    \begin{remark}
        Condition 2 in \eqref{prop:weak_convergence_CU} means that $\mu_a$ is the unique eigenvector corresponding to the eigenvalue one of the adjoint operator of $\cP_t$ as an operator on the dual space $\CU'$ of $\CU$.
    \end{remark}
%
By finally bringing ~\autoref{prop:spectralgap_invertibility} together with~\autoref{prop:weak_convergence_CU}, we can reformulate \autoref{thm:general_response} with the following assumptions, which are similar to assumptions as set out in \cite{HMajda10} but for a generic family of observables. 
%
    \paragraph{Assumption~L}\label{asuL}
    Let $\Ieps := (a_0 - \varepsilon, a_0 + \varepsilon) \subset \R$ be an interval and $\lbrace \{\cP^a_t\}_{t \geq 0} \, : \, a \in \Ieps \rbrace $ a family of Markov semigroups acting on a Hilbert space~$\cH$, with $\mu_a$ an invariant measure of $\{\cP^a_t\}_{t \geq 0}$ for each $a \in \Ieps$.
      Fix some $t > 0$ and let $\cO$ be a Banach space of measurable functions on $\cH$ so that $(1 - \cP_t^{a_0}) \psi + c \in  \cO$ whenever $\psi \in \cO$ and $c \in \R$.
      Finally, for some function $U: \cH \to [1, \infty)$ define the space $\CU := \overline{\cO}$, where the closure is with respect to the norm $\|\cdot\|_U$ (which is assumed to be finite on $\cO$) in the space of measurable functions with finite $\|\cdot \|_U$.
        Suppose the following conditions hold:
      \begin{itemize}
        \item[{\crtcrossreflabel{L1}[L1]}] There exists $\rho<1$ such that 
    \begin{equation*}
        \|\cP^{a_0}_t \varphi - \scp{\varphi, \mu_{a_0}} \|_{\cO} \leq \rho \| \varphi - \scp{\varphi, \mu_{a_0}}\|_{\cO} \fa \varphi\in \cO;
    \end{equation*}
  \item[{\crtcrossreflabel{L2}[L2]}]
for every $\psi \in \cO$ the map $a \mapsto (1 - \cP_t^a)\psi$ has values in $\CU$ and is differentiable at $a = a_0$, and furthermore
            \begin{equation}
            \label{eq:ch5_bound_derivativeII}
                \|D_a \cP_t^a \psi |_{a = a_0} \|_U \leq C \|\psi \|_{\cO};
            \end{equation}
  \item[{\crtcrossreflabel{L3}[L3]}] $\sup_{a \in \Ieps} \scp{ U, \mu_a} = C_{\varepsilon}< \infty$.
              \end{itemize}
              
   \begin{remark}
        Condition \ref{L1} can be replaced by the weaker assumption that one is an isolated eigenvalue of multiplicity one of $\cP_t^{a_0}$ and $\left(\cP_t^{a_0}\right)^*$. This is equivalent to the fact that one is not in the essential spectrum in the sense of Kato, see \cite{Kato}. 
    \end{remark}    
    Now the following version of \cite[Theorem~2.3]{HMajda10} holds.
    \begin{theorem}
    \label{thm:general_response_HM}
        Suppose that \nameref{asuL} holds. 
        Then the map $a \mapsto \scp{ \varphi, \mu_a} $ is differentiable at $a = a_0$ for every $\varphi\in \cO$ and in particular 
        \begin{equation*}
            \left. \dv{}{a} \scp{ \varphi, \mu_a }\right|_{a = a_0} = \scp{ D_a \cP_t^{a}|_{a = a_0} (1 - \cP_t^{a_0})^{-1}(\varphi - \scp{\varphi, \mu_{a_0}}), \mu_{a_0}}. 
        \end{equation*}
    \end{theorem}
    \begin{proof}
    Note that assumption~\ref{L3} together with the requirement that $\|\cdot\|_U$ is well defined on $\cO$ implies the integrability of elements of $\cO$.
    We just need to check that condition~(\ref{itm:weak_convergenceCU2}) in~\autoref{prop:weak_convergence_CU} is satisfied. 
    By~\autoref{prop:spectralgap_invertibility} we have $\Im (1 - \cP_t^{a_0}) = \ker_{\cO} \mu_{a_0}$ with $\ker_{\cO}$ understood in $\cO$.
    Taking closure with respect to $\|\cdot\|_{\CU}$ gives
    $\overline{\Im (1 - \cP_t^{a_0})} = \overline{\ker_{\cO} \mu_{a_0}}$.
    Now let $\xi \in \ker_{\CU} \mu_{a_0}$. 
    Then there exists a sequence $\xi_n \in \tcO$ so that $\xi_n \to \xi$ in $\CU$ by assumption, which implies $\xi_n - \scp{\xi_n, \mu_{a_0}} \to \xi - \scp{\xi, \mu_{a_0}} = \xi$; in other words, we may take $\xi_n \in \ker_{\cO} \mu_{a_0}$.
    This shows that $\overline {\ker_{\CU} \mu_{a_0}} = \ker_{\CU} \mu_{a_0}$, completing the proof.
    \end{proof}
\subsection{Fractional Response}
\label{fractional}
In certain contexts it may not be possible to ensure differentiability of the invariant measure.
In fact, the hardest condition to verify is typically the spectral gap, but also the differentiability of the semigroup may potentially fail. Both of these conditions depend on the family of semigroups and the space of observables considered. 
%
%

%
We consider the same setup as in Section~\ref{subsec:response_setup}.
As in Section~\ref{subsec:linear}, fix a time $t$ and for notational simplicity drop the dependence on time of the semigroup.
We want to show under which conditions the map $a \mapsto \scp{ \varphi, \mu_a}$ is $\alpha$-H\"{o}lder continuous in the interval $I$, that is there are $\alpha \leq 1$ and $C = C(\varphi, \alpha, I)$ such that for any $a_1, a_2 \in I$ we have
\begin{equation*}
      |\scp{ \varphi, \mu_{a_{1}} - \mu_{a_{2}}}| \leq C \,|a_{1} - a_{2}|^\alpha.
    \end{equation*}
Our starting point will again be Equation~\eqref{eq:formal_inversion_(1-P)}.
Let $\CU$ be as in the beginning of \autoref{subsec:linear}.
The motivation to introduce this space is the same as in that section; it allows to require weaker assumptions on the semigroup, rendering them applicable for the models we want to consider.
Using \eqref{eq:formal_inversion_(1-P)} and the definition of $\|\cdot \|_{U}$ in~\eqref{eq: def norm U}, we have
    \begin{equation}
         |\scp{ \varphi, \mu_{a_{1}} - \mu_{a_{2}}}|  = \left|\scp{ (\cP_t^{a_{1}} - \cP_t^{a_{2}})\psi, \mu_{a_{2} }} \right|
         \leq \|\cP_t^{a_{1} }\psi - \cP_t^{a_{2}}\psi\|_U \scp{ U, \mu_{a_{2} }}
         \label{eq:holderU}
    \end{equation}
    whenever $\varphi \in \tcO_{a_2}$.
    Suppose that the map $a\mapsto (1 - \cP_t^a)$ is $\alpha$--H\"{o}lder continuous as a function from $I$ to $L(\cO, \CU)$, the space of linear operators between $\cO$ and $\CU$. That is, there exists $C_1(\alpha, I)$ such that
    \begin{equation*}
      \| \cP_t^{a_{1}} - \cP_t^{a_{2} }\|_{L(\cO, \CU)}
      = \sup_{\|\psi\|_{\cO}\neq 0} \frac{\| \cP_t^{a_{1}}\psi - \cP_t^{a_{2}}\psi\|_{U}}{\|\psi \|_{\cO}} \leq C_1 |a_{1} - a_{2}|^\alpha,
    \end{equation*}
    for any $a_1, a_2 \in I$.
    From \eqref{eq:holderU} it then follows that
    \begin{equation*}
         |\scp{ \varphi, \mu_{a_{1}} - \mu_{a_{2} }}| 
         \leq C_1 \|\psi \|_{\cO}\,  |a_{1} - a_{2} |^\alpha \scp{ U, \mu_{a_{2} }}.
    \end{equation*}
    Finally if we assume $U$ to be such that $\sup_{a \in I} \scp{ U, \mu_{a}} =: C_I < \infty$, then we have
    \begin{equation}
    \label{eq:holderweak2}
        |\scp{ \varphi, \mu_{a_{1}} - \mu_{a_{2} }}| \leq C_2 \|\psi \|_{\cO}\, |a_{1} - a_{2} |^\alpha \fa \varphi\in \tcO_{a_1}
    \end{equation}
    where $C_2 = C_I C_1$, which is the desired result, except that the space of functions $\tcO_{a_1}$ for which this is valid depends on $a_1$.
    This dependence can be removed if we assume~\autoref{prop:spectralgap_invertibility} to hold for all $a \in I$.
    Using also the representation of $\psi$ from that proposition we may rewrite~\eqref{eq:holderweak2} as
    \begin{equation}
    \label{eq:holderweak3}
        |\scp{ \varphi, \mu_{a_{1}} - \mu_{a_{2}}}| \leq \frac{C_2}{1 - \rho(a_1)}\|(\varphi - \scp{ \varphi , \mu_{a_1} } ) \|_{\cO}\, |a_{1} - a_{2}|^\alpha,
    \end{equation}
    for all $\varphi \in \cO$.
     The following assumption parallels \nameref{asuL} and summarises the conditions we have used.
    \paragraph{Assumption~F}\label{asuF}
    Let $I \subset \R$ be an interval and $\lbrace \{\cP^a_t\}_{t \geq 0} \, : \, a \in I \rbrace $ a family of Markov semigroups acting on a Hilbert space~$\cH$, with $\mu_a$ an invariant measure of $\{\cP^a_t\}_{t \geq 0}$ for each $a \in I$.
      Fix some $t > 0$ and let $\cO$ be a Banach space of measurable functions on $\cH$ so that $(1 - \cP_t^{a_1}) \psi + c \in  \cO$ whenever $\psi \in \cO$ and $c \in \R$.
      Finally, for some function $U: \cH \to [1, \infty)$ define the space $\CU := \overline{\cO}$, where the closure is with respect to the norm $\|\cdot\|_U$ 
      (which is assumed to be finite on $\cO$) in the space of measurable functions with finite $\|\cdot \|_U$.
        Suppose the following conditions hold:
    \begin{itemize}
        \item[{\crtcrossreflabel{F1}[F1]}] There exists $\rho < 1$ such that for any $a\in I$
        \begin{equation*}
          \|\cP^{a}_t \varphi - \scp{ \varphi, \mu_{a} } \|_{\cO}
          \leq \rho \| \varphi - \scp{ \varphi , \mu_{a} } \|_{\cO} \fa \varphi \in \cO.
        \end{equation*}
    \item[{\crtcrossreflabel{F2}[F2]}]
 The map $a \mapsto (1 - \cP_t^a)$ is  $\alpha$--H\"{o}lder continuous for some $\alpha\leq 1$, as a function from $I$ to $L(\cO, \CU)$, that is there exists $C = C(\alpha, I)$ such that
    \begin{equation*}
      \|\cP_t^{a_{ 1}} -  \cP_t^{a_{2}}\|_{L(\cO, \CU)}
      \leq C \, |a_{1} - a_{2} |^{\alpha} .
    \end{equation*}
    \item[{\crtcrossreflabel{F3}[F3]}] $\sup_{a \in I} \scp{ U, \mu_a} = C_I < \infty$.
    \end{itemize}
    Then we have shown the following result:
    \begin{theorem}
    \label{thm:general_holder}
        Suppose that \nameref{asuF} holds. 
        Then the invariant measure $\mu_a$, as operator on $\cO$, is $\alpha$--H\"{o}lder continuous in $a \in I$, namely there exists $C$ depending on $\varphi$ such that for all $a_{1}, a_{2} \in I$ 
        \begin{equation*}
             |\scp{ \varphi, \mu_{a_{1}} - \mu_{a_{2}}}| \leq C |a_{1} - a_{2}|^\alpha \fa \varphi \in \cO.
        \end{equation*}
    \end{theorem}
    \begin{proof}
      We have seen that under the conditions of the theorem, Equation~\eqref{eq:holderweak3} holds so it remains to establish the bound
      \begin{equation*}
        \frac{C_2}{1 - \rho(a_1)}\|(\varphi - \scp{ \varphi , \mu_{a_1} } ) \|_{\cO} \leq C
      \end{equation*}
      for some $C$ independent of $a_1$.
      By condition~\ref{F1} we may bound $\rho(a_1)$ by $\rho$ uniformly.
      Furthermore, $\|(\varphi - \scp{ \varphi , \mu_{a_1} } )        
      \|_{\cO} \leq \|\varphi\|_{\cO} + |\scp{ \varphi , \mu_{a_1} }| \|1\|_{\cO}$, where $1$ here means the function equal to~1 everywhere.
      We have also used that $|\scp{ \varphi , \mu_{a_1} }| \leq \|\varphi\|_U C_I$ on a number of occasions already.
      These facts imply that we may take $C = \frac{C_2}{1 - \rho} (\|\varphi\|_{\cO} + C_I \|1\|_{\cO} \|\varphi\|_U )$.
      \end{proof}
        For the remainder of this section we will discuss examples of spaces $(\cO, \|\cdot\|_{\cO})$ that are appropriate for the applications considered in subsequent sections.
          It turns out that rather than {\em norms} it is practical to consider {\em semi}--norms which are zero on constant functions, or in other words which are norms modulo the constant functions. 
          For instance $\|(\varphi - \scp{ \varphi , \mu_{a_1} } ) \|_{\cO} = \|\varphi \|_{\cO}$ in case that $\|\cdot\|_{\cO}$ is such a seminorm, and the constant $C$ in the proof of Theorem~\ref{thm:general_holder} simply reads as $C = \frac{C_2}{1 - \rho} \|\varphi \|_{\cO}$.
          If $\|\cdot\|$ is a norm modulo the constant functions, there are many ways to turn $\|\cdot\|$ into a proper norm; the mapping $\varphi \to |\varphi(x_0)| + \|\varphi\|$ for instance provides a norm, with $x_0$ an arbitrary element of $\cH$.
          In the present paper we will use for $\cO$ the space $\Cd$ defined as the space of measurable functions on $\cH$ with finite Lipschitz seminorm
    \begin{equation*}
        \| \varphi \|_{\td} = \sup_{x\neq y} \frac{|\varphi(x) - \varphi(y)|}{\td(x,y)}.
    \end{equation*}
    Here $\td$ is a \textit{semimetric} on $\cH$, namely $\td: \cH\times \cH \to \R_+$ is symmetric, lower semi-continuous and such that $\td(x,y) = 0 \Leftrightarrow x = y $. When the symmetry fails, we refer to $\td$ as a \textit{premetric}. 
   Consider the Wasserstein semimetric associated to $\td$, namely
   \begin{equation*}
        W_{\td}(\mu_1, \mu_2) := \inf_{\Gamma \in  \mathcal{C}(\mu_1, \mu_2)} \int \td(x, y) \; \Gamma(dx, dy) ,
    \end{equation*}
    where $\mu_1, \mu_2\in \cM_1(\cH)$ and $\mathcal{C}(\mu_1, \mu_2)$ is the set of all couplings of $\mu_1, \mu_2$.
    The inequality 
    \begin{equation}
    \label{eq: sup <f, mu> leq W_d}
       \sup_{\|\varphi\|_{\td}\leq 1} \left| \langle \varphi, \mu_1 \rangle -  \langle \varphi, \mu_2 \rangle\right| \leq  W_{\td}(\mu_1, \mu_2),
    \end{equation}
  follows directly from the definition of the Wasserstein semimetric.
  When $\td$ is a metric, $W_{\td}$ is itself a metric and inequality~\eqref{eq: sup <f, mu> leq W_d} turns out to be an equality (Kantorovich-Rubinstein formula).
  When $\td$ is merely a semimetric, like in the framework we are about to work with, we can still provide a sufficient condition for \ref{F2} in terms of the Wasserstein semimetric $W_{\td}$.
    In fact 
    \begin{align}
            \left| (\cP^{a_1}_t - \cP^{a_2}_t)\psi(x) \right| &=
           \left| \langle \psi , P^{a_1}_t(x, \cdot)\rangle - \langle \psi , P^{a_2}_t(x, \cdot)\right| \nonumber \\
            &\leq \| \psi\|_{\td} \, W_{\td}(P^{a_1}_t(x, \cdot), P^{a_2}_t(x, \cdot)).\label{eq:semigroupwass}
    \end{align}
    Then, if $a \mapsto P_t^{a}(x, \cdot)$ is H\"{o}lder continuous in the Wasserstein semimetric $W_{\td}$ in the sense that
    \begin{align*}
         W_{\td}(P_t^{a_1}(x, \cdot), P_t^{a_2}(x, \cdot)) \leq |a_1 - a_2|^\alpha C(t)U(x) \fa x\in \cH
    \end{align*}
    with $U$ satisfying \ref{F3}, we get from \eqref{eq:semigroupwass}
    \begin{equation*}
        \|(\cP^{a_1}_t - \cP^{a_2}_t)\psi \|_{U}\leq |a_1 - a_2|^\alpha \| \psi\|_{\td} C(t)
    \end{equation*}
    as desired.
    We have then shown the following result:
    \begin{corollary}
    \label{thm:generalmuHolder}
       Let $\cP_t^a$ be a family of Markov semigroups acting on $(\Cd, \|\cdot\|_{\td})$ with transition probabilities $P_t^a$. Assume \ref{F1}, \ref{F3} hold, as well as condition
       \begin{itemize}
          \item[{\crtcrossreflabel{F2'}[F2']}] there exists a positive function $C = C(\alpha, t, \varepsilon)$ such that 
        \begin{align*}
                 W_{\td}(P_t^{a_1}(x, \cdot), P_t^{a_2}(x, \cdot)) \leq |a_1 - a_2|^\alpha C(t)U(x).
        \end{align*}
       \end{itemize}
      Then \ref{F2} holds and thus the conclusions of \autoref{thm:general_holder}. 
     \end{corollary}
 \begin{remark}\label{rmk:weaktriangular}
 It may appear that the result is covered in \cite[section~4.1]{HMattSch11} but as $\td$ is only a semimetric the two results are not directly linked in general. In \cite{HMattSch11} the authors show that if $\td$ satisfies a weak form of the triangular inequality, if conditions \ref{F2'} and \ref{F3} are satisfied and if there exists $\rho<1$ such that 
 \begin{equation}
 \label{eq:contraction}
     W_{\td}(\cP_t^{a_1} \nu_1, \cP_t^{a_1}\nu_2)\leq \rho W_{\td}(\nu_1, \nu_2), \fa \nu_1, \nu_2\in \cM_1(\cH),
 \end{equation}
 then 
    \begin{equation*}
        W_{\td}(\mu_{a_1}, \mu_{a_2}) \leq 2|a_1 - a_2| C(t) C_I.
    \end{equation*}
 Thanks to \eqref{eq: sup <f, mu> leq W_d}, this result will imply our result \autoref{thm:generalmuHolder}.
 However condition \eqref{eq:contraction} and the spectral gap result \ref{F1}
 do not imply one another. 
    \end{remark}

%
\section{Methodology for SPDEs}
\label{sec:response_SPDE}
%

In this section we aim to give sufficient conditions suitable for dissipative SPDEs which imply the conditions \nameref{asuL} considered in \autoref{asuL} to obtain Linear Response and \nameref{asuF} considered in \autoref{asuF} to obtain Fractional response. First, we give a precise description of the dissipative SPDEs.
 Let $(\cH, |\cdot|)$ and $(\cV, \| \cdot \|)$ be Hilbert spaces with $\cV \subset \subset \cH$ (i.e.\ $\cV$ is compactly contained in $\cH$) and $\cV$ is dense in $\cH$.
 Further, w.l.o.g. $\|v\| \geq |v|$.
 This implies that $\cH = \cH' \subset \subset \cV'$, that $|v| \geq \|v\|_{\cV'}$, and that $\cH$ is dense in $\cV'$.
 Consider the stochastic equation 
  \begin{equation}
  \label{eq:ch5generalsde}
      dX +  AX\, dt= \left(  F(X) + f(a)  \right) \, dt + dW, \quad X(0) = x
  \end{equation}
  where $A:\cV \to \cV'$ is a nonnegative linear operator, $f(a)\in \cV'$ is a deterministic forcing dependent on the parameter $a \in I$, $W$ is a Wiener process on the probability space $(\Omega, \mathcal{F}, \bP)$ with values in $\cH$ and trace class covariance operator $Q:\cH \to \cH$, and $F:\cV \to \cV'$ a nonlinear continuous function. We consider Equation~\eqref{eq:ch5generalsde} as an equation in $\cV'$, and we assume that there exists a unique solution for any initial condition $X(0) = x\in \cH$. 
  For any $a \in I$ we assume the solution to be in $C([0,T]; \cH)\cap L^2(0,T; \cV)$ for all $T >0$ and a.a.~$\omega\in \Omega$, to be continuous with respect to the initial condition (as a function into $C([0,T]; \cH)$), and depend continuously on the forcing $f(a)$ (as a function $\cV'\to \cH$). 
  We define the Markov semigroup $\cP_t$ on $B_b(\cH)$ by
  \( \cP^a_t \varphi (x) = \E \, \varphi (X_a(t, \cdot; x))\)
  and denote the associated transition probabilities as $P^a_t(x, \cdot)$. Given the regularity of the solutions the associated semigroup is Feller and we assume it admits an invariant measure $\mu_a$.
  
  Denote with $L_2(\cH)$ the space of Hilbert--Schmidt operators, so that $Q^{1/2}\in L_2(\cH)$, since $Q$ is trace class, and define $L_2^0 (\cH)$ to be the space of elements $T \in L_2(\cH)$ such that
  \begin{equation*}
      \| T \|_{L_2^0}^2 = \sum_{k\in \N} |T Q^{1/2} e_k|^2
  \end{equation*}
is finite, where $\lbrace e_k, \, k\in \N \rbrace$ is an orthonormal eigenbasis of $\cH$. 
 
 In the next subsections we will discuss conditions under which \nameref{asuL} holds. 
 Most importantly, we will have to require that the forcing $f(a)$ is in the range of the noise covariance $Q$ in order to show that the semigroup is differentiable with respect to the parameter.
 We will then move to fractional response and provide verifiable conditions suitable for SPDEs of the form \eqref{eq:ch5generalsde}, which do not require that $f(a)$ lies in the range of $Q$. 
 %
 \subsection{Spectral gap}
 
   
We saw in \autoref{sec:response_method} that for both differentiability and H\"{o}lder continuity we want to show the operator $(1 - \cP_t^{a})$ to be invertible on an appropriate set of observables. By \autoref{prop:spectralgap_invertibility} having a spectral gap for $\lbrace \cP^a_t\rbrace_{t\geq 0}$ is a sufficient condition to ensure invertibility. 
We discuss two examples of appropriate spaces $\cO$ over which we have the spectral gap.

    In the original work of Hairer and Majda \cite{HMajda10} the set $\cO$ is the closure of the space of cylinder functions $C^\infty_0(\cH)$ under the following norm 
    \begin{equation}
    \label{eq:V,W norm}
        \| \varphi \|_{1; V,W} = \sup_{x\in \cH} \left( \frac{|\varphi(x)|}{V(x)} +  \frac{\|D\varphi (x)\|}{W(x)}\right),
    \end{equation}
    where $V, W: \cH \to [1, \infty)$ are two continuous functions.
    Their choice for $\cO$ fits with the results in \cite{HMatt08}; indeed that work ensures the spectral gap condition in the norm \eqref{eq:V,W norm} for a large class of hypoelliptic diffusions and in particular for the 2D stochastic Navier--Stokes equations with highly degenerate noise. 
    
    As the authors in \cite{HMajda10} observe, it can be proved that if we quotient this space by the space of all constant functions, then there is a distance function $d_{V,W}$ such that this norm is equivalent to the Lipschitz norm corresponding to $d_{V,W}$ i.e.
    \begin{equation*}
        \| \varphi \|_{1; V,W} = \sup_{x\neq y} \frac{|\varphi(x) - \varphi(y)|}{d_{V,W}(x,y)}.
    \end{equation*}
    Then this choice of space fits also in the framework developed for fractional response in \autoref{thm:generalmuHolder} as long as conditions \ref{F2'} and \ref{F3} hold. 
    
    Our analysis will be based on a different spectral gap result, which makes use of another space of observables, or more precisely of a different metric in the definition of the Lipschitz norm. 
   Indeed, a generalised form of Harris' theorem (\cite{HMattSch11,butkovsky2020}) ensures the semigroup $\cP_t^a$ exhibits a spectral gap in the Lipschitz seminorm corresponding to a semimetric $\td$ of the form 
     \begin{equation}
        \label{ch5:def td}
        \td(x,y) = \sqrt{d(x,y) (1 + V(x) + V(y))},
    \end{equation}
     where $d\leq 1$ is another semimetric on $\cH$ (satisfying appropriate conditions) and $V:\cH \to [0, \infty]$ is a Lyapunov function for $\cP_t^a$, namely it satisfies the following:
    \begin{definition}[\cite{HMattSch11}]
   \label{def:LyapunovfctHMS}
        A measurable function $V: \cH \to \R_+$  is called \emph{Lyapunov function} for  a Markov semigroup $\{\cP_t, t \geq 0\}$ if there exist positive constants $C$, $\gamma$, $K$ such that 
        \begin{equation}
        \label{eq:ch4lyapunovcondition}
            \cP_t V(x) \leq C e^{-\gamma t} V(x) + K \fa x\in \cH, \, t\geq 0. 
        \end{equation}
   \end{definition}
     As we want the semimetric $\td$ in definition~\eqref{ch5:def td} to be independent of the choice of the parameter $a$, we assume that $V$ is a Lyapunov function for $\cP_t^a$ for any choice of the parameter $a$, namely that for any $a \in \Ieps$ there exists $K_a, C_a, \gamma_a$ such that
       \begin{equation*}
       \cP_t^a V(x) \leq C_a e^{- \gamma_a t}V(x) + K_a.
   \end{equation*}
   Then we define the observable space $(\cO, \| \cdot \|_{\cO})$ as $\Cd$, the space of measurable functions $\varphi$ with
   \begin{equation*}
        \| \varphi \|_{\td} = \sup_{x\neq y} \frac{|\varphi(x) - \varphi(y)|}{\td(x,y)} = \sup_{x\neq y} \frac{|\varphi(x) - \varphi(y)|}{ \sqrt{d(x,y) (1 + V(x) + V(y))}} <\infty.
   \end{equation*}
    \begin{remark}\label{rmk:quotientspace}
       As $\|\varphi\|_{\td} = 0$ only implies $\varphi$ is a constant function, $\|\cdot \|_{\td}$ is only a seminorm. However it is enough to quotient $\Cd$ by the space of all constant functions in order to make sure it is a Banach space as desired. 
      Equivalently we could have changed the seminorm to $\| \varphi\|_{\td} + |\varphi(x_0)|$, for $x_0\in \cH$, to ensure it is a norm, but we chose not do do so for simplicity's sake, and for better comparison with the available literature.
      \end{remark} 
    We further need a space $\CU$ of measurable functions which is a Banach space under the norm \eqref{eq: def norm U} with  $U: \cH \to [1, \infty)$ and such that $\Cd \subset \CU$. Let $U$ be such that $U \geq \sqrt{1 + V}$ and define $\CU$ as the closure of $\Cd$ in the space of all measurable functions with finite $\|\cdot \|_U$ norm.
    To ensure that $\CU$ is well defined it is sufficient to prove that there exists a constant $k >0$ such that 
     \begin{equation*}
         \| \varphi \|_U \leq k \| \varphi \|_{\td} \fa \varphi \in \Cd.
     \end{equation*}
    This follows directly (see \cite[Proposition~5.1.6.]{thesis}) from the definition \eqref{ch5:def td} of $\td$, the requirement $U \geq \sqrt{1 + V}$ and the integrability of $V$ with respect to any invariant measure $\mu_{a_0}$.
        Indeed it can be shown (see \cite[Lemma~4.1]{butkovsky2014}) that a Lyapunov function is integrable with respect to any invariant measure of the semigroup.
       In our context this implies that $V$ will be integrable with respect to any invariant measure  $\mu_a$ of $\{\cP_t^a\}$ for any $a\in \Ieps$.
       Then we also have an explicit upper bound for its integral which will be useful later on to ensure the integrability of $U$ as well: by the invariance of $\mu_a$ with respect to $\lbrace \cP^a_t, t\geq 0\rbrace$ and by \eqref{eq:ch4lyapunovcondition} we have that 
    \begin{equation*}
        \langle V , \, \mu_a \rangle =  \langle \cP^a_t V , \, \mu_a\rangle \leq  C_a e^{-\gamma_a t} \langle V , \, \mu_a \rangle  +  K_a,
    \end{equation*}
    and therefore
    \begin{equation}
    \label{eq:ch4Vintegrable}
         \langle V , \, \mu_a \rangle \leq \frac{K_a}{(1 - C_a e^{-\gamma_a t})} .
    \end{equation}
Recently \cite{carigi2022exponential} provided a set of verifiable conditions (see \nameref{asuA} below) for nonlinear dissipative SPDEs as \eqref{eq:ch5generalsde}, inspired by the results in \cite{butkovsky2020}, which are sufficient to get a spectral gap in $\Cd$ where $\td$ is defined in \eqref{ed:deftd} below. Let us set the notations to state these conditions and explain the main idea of the proof.
    Let $Y_a$ be the solution of 
    \begin{align*}
        dY_a + AY_a \, dt = \left(  F(Y_a) + f(a)\right)\, dt  + d W_Y, \quad Y_a(0) = y\neq x,
    \end{align*}
    that is, $Y_a$ satisfies the same equation~\eqref{eq:ch5generalsde} as $X_a$ but with a different initial condition and a different Wiener process $W_Y$.
    Therefore $\Law Y_a(t) = P^{a}_t(y, \cdot)$.
    We want to show that the distance in the Wasserstein semimetric between $\Law X_a(t)$ and $\Law Y_a(t)$ becomes small for large enough $t$, as this will then imply the desired results in \autoref{thm:expconvergence}.
    %
    %
    To this end consider an intermediate process
    \begin{equation}
    \label{eq:ch4equationYtilde}
        d\tY_a + A\tY_a \, dt = \left(  F(\tY_a) + f(a)  + G(X_a, \tY_a)\right) \, dt + dW, \quad \tY_a(0) = y,
    \end{equation}
    which we assume has a solution of the type described in connection with Equation~\eqref{eq:ch5generalsde} at the beginning of this section, where $G$ is a so-called control function.
    This control is chosen so that the distance between $X_a(t)$ and $\tY_a(t)$ in the semimetric $\td$ gets small for large enough $t$ and the Wasserstein semimetric between their laws gets small as well.
    However, $\Law \tY_a(t) \neq P^{a}_t(y, \cdot)$ due to the presence of the extra control term.
    If the control function $G$ can be taken in a suitable finite~dimensional subspace of $\cH$, then the distance between $\Law \tY_a(t)$ and $P^{a}_t(y, \cdot)$  can be investigated using the Girsanov theorem (for details see \cite{butkovsky2020,carigi2022exponential}).
    
    Under the following conditions the argument sketched here can be made rigorous:
    %
    %
    %
    
    %
    %
  
    %
    \paragraph{Assumption~A}\label{asuA} 
    Let $\cH_n$ be an $n$-dimensional subspace of $\cH$ with $\Pi_n$ the orthogonal projection onto $\cH_n$.
    The covariance operator $Q$ commutes with $\Pi_n$, and furthermore $Q_n := \Pi_n Q$ is invertible on $\cH_n$.
     Let $\Ieps = (a_0 - \varepsilon, a_0 + \varepsilon) \subset \R$ and fix $a\in \Ieps$. Given a solution $t \to X_{a}(t)$ of \eqref{eq:ch5generalsde}, there exists a measurable function $G: \cH\times \cH \to \cH_n$ such that the controlled equation \eqref{eq:ch4equationYtilde} has a unique solution $\tY_{a}$ in the sense described at the beginning of the current section. 
    %
    In addition, we require the following:
    \begin{itemize}
        \item[{\crtcrossreflabel{A1}[ref:A1]}] There exist $\kappa_0>0$ and $\kappa_1\geq 0$ independent of $a$ such that for all $t\geq 0$
        \begin{equation*}
            |X_{a}(t) - \tY_{a}(t)|^2\leq |x - y|^2 \exp(- \kappa_0 t + \kappa_1 \int_0^t \| X_{a}(s) \|^2 \, ds). 
        \end{equation*}
        \item[{\crtcrossreflabel{A2}[ref:A2]}] There exist $\kappa_2> 0$, $\kappa_\varepsilon\geq 0$ independent of $a$, and for each $\gamma>0$ a random variable $\Xi^a_\gamma$ such that 
        \begin{equation*}
            |X_{a}(t)|^2 + \kappa_2 \int_0^t  \| X_{a}(s) \|^2 \, ds \leq  |x|^2 + \kappa_\varepsilon t + \Xi^{a}_\gamma ,\quad t\geq 0
        \end{equation*}
        with $\kappa_0\kappa_2 > \kappa_1 \kappa_{\varepsilon} $ and 
        \begin{equation}
        \label{eq:A2 martingale estimate}
            \bP(\Xi^{a}_\gamma \geq R) \leq e^{-2\gamma R}, \fa R \geq 0.
        \end{equation}
        \item[{\crtcrossreflabel{A3}[ref:A3]}] There exists a positive constant $c > 0$, independent of $a$, such that for each $t\geq 0$ and $s\in [0,t]$
        \begin{equation*}
            |G(X_a(s), \tY_a(s))|^2 \leq c | X_a(s) - \tY_a(s) |^2.
        \end{equation*}
        \item[{\crtcrossreflabel{A4}[ref:A4]}] There exists a measurable function $V: \cH \to \R_+$, independent of $a$, such that for some $\gamma_\varepsilon, K_\varepsilon>0$
        so that the following estimate holds for $t\geq s\geq 0$ and all $a \in I$:
        \begin{equation*}
            \E V(X_a(t)) \leq \E V(X_a(s)) + \int_s^t \left( -\gamma_\varepsilon \E V(X_a(\tau)) + K_\varepsilon \right)\; d\tau .
        \end{equation*}
        Furthermore, for any $M>0$ the function $x\mapsto |x|^2$ is bounded on the sublevel sets $\lbrace V \leq M \rbrace$.
    \end{itemize}
      
        As a direct consequence of \ref{ref:A4}, the function $V$ is a Lyapunov function for the semigroup $\cP_t^a$ with
        \begin{equation}
        \label{eq:ourLyapunov}
            \cP_t^a V(x) \leq V(x) e^{- t\gamma_\varepsilon} +\tfrac{K_\varepsilon}{\gamma_\varepsilon} \fa x\in \cH,\, t\geq 0.
        \end{equation}
    Furthermore, as seen in \cite{carigi2022exponential, butkovsky2020}, \nameref{asuA} 
    is sufficient to show spectral gap on a space $\Cd$ of observables for a semimetric $\td$ defined as follows.
    Given $\kappa_1, \kappa_2$ and $\gamma$ as in \ref{ref:A1} and \ref{ref:A2}, set
        \begin{equation}
        \label{eq:QGupsilonalpha0}
         \upsilon = \frac{\kappa_1}{\kappa_2} \quand \alpha_0 =  \frac{1}{2}\wedge \frac{2\gamma}{\upsilon + 2\gamma}.
        \end{equation} 
    For $\alpha\in (0, \alpha_0)$ define the premetric $\theta_\alpha$ 
    \begin{equation}
    \label{eq:ch4 theta(u,v)}
        \theta_\alpha(x,y) := |x - y|^{2\alpha} e^{\alpha \upsilon |x|^2}
    \end{equation}
      and, given $N\in \N$, define the semimetric $d_N$ as 
    \begin{equation}
        \label{eq:ch4defdN}
        d_N (x,y) := N \theta_\alpha(x,y) \wedge N\theta_\alpha(y,x) \wedge 1 
    \end{equation}
    and finally the semimetric considered is given by
    \begin{equation}
    \label{ed:deftd}
    \td(x,y) := \sqrt{d_N(x,y)(1 + V(x) + V(y))}
     \end{equation}
     for a suitably chosen $N$. 
    %
   Then the following result holds:
    \begin{theorem}[{\cite{carigi2022exponential}}]
    \label{thm:expconvergence}
    Consider $X_a(t)$ solution of \eqref{eq:ch5generalsde} with associated Markov semigroup $\cP^{a}_t$ and invariant measure $\mu_a$. Suppose \nameref{asuA} holds and define the space $\Cd$ with $\td$ as in \eqref{ed:deftd}.
        Then the following holds: 
        \begin{enumerate}
            \item \emph{Exponential stability}: there exists positive constants $r,\, C,\, t_0$, independent of $a$, such that 
        \begin{equation*}
            W_{\td}(P^a_t(x, \cdot), \mu_a) \leq C(1 + V(x)) e^{-r t} \fa x\in \cH, t\geq t_0. 
        \end{equation*}
            \item \emph{Spectral gap}: There exists and $\rho<1$ such that 
        \begin{equation*}
            \| \cP^{a}_{t} \varphi  - \langle \varphi , \mu_a \rangle  \|_{\td} \leq  \rho \| \varphi - \langle \varphi, \mu_a \rangle  \|_{\td} \fa \varphi\in \Cd,\, t \geq 0.
        \end{equation*}
    \end{enumerate}
      
    \end{theorem}
    

By definition, the semimetric $d_N$ is comparable to the $\alpha_0$--power of the original metric on the space, with $\alpha_0$ as in \eqref{eq:QGupsilonalpha0}. Therefore, \autoref{thm:expconvergence} is giving us that, $\cP^a_t$ exhibits a spectral gap on the set of observables which are $\alpha_0$--H\"{o}lder continuous over the level sets of the Lyapunov function $V$. 
%



\subsection{Linear Response}
\label{subsec:SPDElinear}
   For the space of observables $(\Cd, \| \cdot \|_{\td})$ we want to show under which condition the model introduced in \eqref{eq:ch5generalsde} exhibits linear response with respect to the parameter $a$. We will do that by giving sufficient conditions for \nameref{asuL} to hold. 
   First of all we note that \autoref{thm:expconvergence} ensures that Assumption~\ref{L1} holds. 
     As observed above, the observables in $\Cd$ are not differentiable but locally H\"{o}lder with exponent $\alpha < \frac{1}{2}$ with respect to the norm $|\cdot|$ on $\cH$.
     The lack of differentiability of the observables effectively requires some regularization property of the semigroup to ensure \ref{L2} holds.
%
    %
   \begin{theorem}
   \label{thm:diffsemigroup}
        Set $\Ieps = (a_0 - \varepsilon, a_0 + \varepsilon)\subset \R$, $\varepsilon>0$.  Let $X_a(t)$ be the solution of \eqref{eq:ch5generalsde}
        with the map $a\mapsto f(a)$ continuously differentiable as a map from $\Ieps$ into $\range Q$ with
        \begin{equation}
        \label{eq:conditionDf}
            \sup_{a\in \Ieps} |D_a Q^{-1/2}f(a)| < \infty,
        \end{equation}
        and let $\left(\cP_t^a\right)_{t\geq 0}$ be the associated semigroup. 
        Let $V$ be a Lyapunov function for $\left(\cP_t^a\right)_{t\geq 0}$ for any $a\in \R$ and consider a function $U: \cH \to [1, \infty)$ such that $U\geq \sqrt{1+ V}$. 
        Then Assumption~\ref{L2} holds.
        Furthermore the derivative $\left.\left(D_a \cP_t^a\right)\right|_{a = a_0}$ has the explicit formulation 
       \begin{equation*}
            \left. \left( D_a \cP^a \psi \right) \right|_{a = a_0}(x) =  \E\,  \left [\,\psi(X_{a_0}(t, x))  (Q^{-1} \left(D_a f\right)|_{a = a_0}, W(t))\, \right].
        \end{equation*}
    \end{theorem}
    \begin{proof}
        We want to ensure that there exists a function $\left.\left(D_a \cP_t^a \psi \right)\right|_{a = a_0} \in \CU$ such that 
        \begin{equation*}
            \lim_{a \to a_0} \frac{\| \cP^a_t\psi - \cP^{a_0}_t \psi - (a - a_0) \left.\left(D_a \cP_t^a \psi \right)\right|_{a = a_0} \|_U}{|a - a_0|} =0 
        \end{equation*}
        for any $\psi\in \Cd$. 
        Equivalently, given the definition of $\| \cdot \|_U$, we have to show that
        \begin{equation*}
            \lim_{a \to a_0} \sup_{x\in \cH} \frac{1}{U(x)}\left| \E \frac{\psi(X_a(t, x)) - \psi(X_{a_0}(t, x)) - (a- a_0)\left.\left(D_a \cP_t^a \psi\right)\right|_{a = a_0} }{a - a_0} \right| =0.
        \end{equation*}
        Since $\psi$ is not differentiable this will not follow directly from the differentiability of the solution $X$ with respect to the parameter. To go around this problem we introduce the It\^{o} process 
        \begin{equation*}
            d\tilde{W}^a : = \left(f(a) - f(a_0) \right)\, dt + d W.
        \end{equation*}
       By \eqref{eq:conditionDf} and the mean value theorem, the integral
        \begin{equation*}
            \int_0^T |Q^{-1/2}(f(a) - f(a_0))|^2\, ds \leq T\sup_{\tilde{a}\in \Ieps} | D_a Q^{-1/2} f(\tilde{a})|^2
        \end{equation*}
        is well defined.
        Then, by Girsanov's theorem in Hilbert spaces (see e.g.~\cite[Theorem~10.14]{DaPratoZab2014}), the process $\tilde{W}$ is a $Q$-Wiener process on $(\Omega, \tilde{\bP})$ where $\tilde{\bP}$ is a probability measure absolutely continuous with respect to $\bP$ with density 
        \begin{equation}
        \label{eq: density girsanov}
           \frac{ d\tilde{\bP} }{d\bP}(t,a) = \exp( M^a(t)   -  \tfrac{1}{2}\langle M^a\rangle_t )
        \end{equation} 
        where 
        \begin{equation}
        \begin{split}
            M^a(t) = \left(  Q^{-1}  ( f(a) - f(a_0)), W(t) \right),\\
            \langle M^a\rangle_t = t | Q^{-1/2}  ( f(a) - f(a_0))|^2.
        \end{split}\label{eq:def_M(t,a)}
        \end{equation}
        %
        %
        By definition of $\Tilde{W}$, the solution $\tX_{a_0}$ of 
        \begin{equation*}
            d\tX_{a_0} + A\tX_{a_0} \, dt= (F(\tX_{a_0})+ f(a_0 ))\, dt + d\tilde{W}^{a_0} , \quad \tX_{a_0}(0) = x
        \end{equation*}
        is equivalent to the solution $X_a$ of \eqref{eq:ch5generalsde} and $ \E \, \psi(X_a(t, x))  =  \tilde{\E} \, \psi(\tX_{a_0}(t, x))  $.
        Moreover it follows from \eqref{eq: density girsanov} that
        \begin{equation}
        \label{eq: girsanov expectation}
            \tilde{\E} \, \psi(\tX_{a_0}(t, x)) = \E \left[ \psi(X_{a_0}(t, x))  \frac{ d\tilde{\bP} }{d\bP}(t,a) \right] .
        \end{equation}
        Now, taking \textit{formally} the derivative of \eqref{eq: girsanov expectation} at $a_0$ we have 
        %
         \begin{equation*}
            \left. D_a  \E \, \psi(X_a(t, x)) \right|_{a = a_0} = \E \left[ \psi(X_{a_0}(t, x))  D_a\frac{ d\tilde{\bP} }{d\bP}(t,a_0)\right].
        \end{equation*}
        We have to make sure this candidate is indeed the derivative of $\cP_t^a\psi$ at $a_0$, namely to ensure that
        \begin{equation}\label{eq:defmain}
            \sup_{x} \frac{1}{U(x)}\left| \E \left[ \frac{\psi(X_a(t, x)) - \psi(X_{a_0}(t, x))}{a - a_0} - \psi(X_{a_0}(t, x)) D_a\frac{ d\tilde{\bP} }{d\bP}(t,a_0) \right] \right| 
        \end{equation}
        converges to zero when $a$ approaches $a_0$.
        Let us start then by defining the process
       \begin{align}
           m^a(t):= \left( D_a Q^{-1} f (a), W(t)\right), \fa a\in \Ieps \label{def:m(t,a)}
       \end{align}
       which has mean zero and 
       \begin{equation*}
           \E |m^a(t)|^2 = t |D_a Q^{-1/2} f(a)|^2.
       \end{equation*}
       By the differentiability of the exponential function it follows that, almost surely, the density function \eqref{eq: density girsanov} is differentiable and the derivative at $a_0$ is
        \begin{equation*}
            D_a\frac{ d\tilde{\bP} }{d\bP}(t,a_0) = \left(  D_a Q^{-1} f(a_0), W(t)\right) = m^{a_0}(t).
        \end{equation*}
       Since $\E\, m^{a_0}(t) = 0$, we see that \eqref{eq:defmain} does not change if we consider $\psi + k$ with $k$ being a constant, which can depend on the initial condition $x$. 
        Therefore using \eqref{eq: girsanov expectation} we have, if we choose $k = -\psi(x)$,
        \begin{align*}
            \E \left[ \frac{\psi(X_a(t, x))  -\psi(x) + \psi(x) - \psi(X_{a_0}(t, x))}{a - a_0} - \left( \psi(X_{a_0}(t, x)) - \psi(x) \right)  m^{a_0}(t)\right]   = \\
        \E \left[ (\psi(X_{a_0}(t, x))-\psi(x)) S(t,a) \right],
        \end{align*}
        where
        \begin{equation}
            \label{eq: def S(h)}
            S(t,a) =  \tfrac{1}{a - a_0} \left(\frac{ d\tilde{\bP} }{d\bP}(t,a) - 1 \right) -  m^{a_0}(t).
        \end{equation}
        %
        Since $\|\psi\|_{\td}<\infty$, for any fixed $t>0$, we have, by definition \eqref{ch5:def td} of $\td$ 
        \begin{equation}
        \label{eq535}
            |\psi(X_{a_0}(t)) - \psi(x)|  \leq \| \psi\|_{\td} \, \td(X_{a_0}(t), x) \leq \| \psi\|_{\td}(1 + V(X_{a_0}(t)) + V(x))^{1/2}.
        \end{equation}
        Therefore, by Cauchy-Schwartz inequality
        \begin{equation*}
            \E \left[ | \psi(X_{a_0}(t, x))- \psi(x)| \left| S(t,a) \right| \right] \leq \| \psi\|_{\td}\,\sqrt{\E (1 + V(X_{a_0}(t)) + V(x))  } \sqrt{ \E \left| S(t,a)\right|^2 },
        \end{equation*}
        so that 
        \begin{multline}
        \label{eq:longone}
            \lim_{a \to a_0} \sup_{x\in \cH} \frac{1}{U(x)} \E | \psi(X_{a_0}(t, x)) - \psi(x)| \left| S(t,a) \right| \leq  \\ \| \psi\|_{\td}\,\underbrace{\sup_{x\in \cH} \frac{\left( 1 + V(x) + \E V(X_{a_0}(t))  \right)^{1/2}}{U(x)}}_{(I)} \underbrace{\lim_{a \to a_0} \left( \E \left| S(t,a)\right|^2 \right)^{1/2}}_{(II)}.
        \end{multline}
        Let us examine the terms on the right hand side of this expression. 
        
            \paragraph{(I)} By \autoref{def:LyapunovfctHMS} of the Lyapunov function $V$ there exist positive constants $C_{a_0}$, $\gamma_{a_0}$, $K_{a_0}$ such that
            \begin{equation}
            \label{eq:lyapunovlinear}
            \cP_t^{a_0} V(x) \leq C_{a_0}e^{-\gamma_{a_0} t}V(x) + K_{a_0}.
            \end{equation}
            It follows that (I) can be bounded above by
            \begin{equation*}
            \sup_{x\in \cH}\frac{\left(1 + K_{a_0}+ V(x)(1 + C_{a_0} e^{-\gamma_{a_0} t})\right)^{1/2}}{U(x)} .
            \end{equation*}
            Then, since $U(x)\geq \left(1+ V(x)\right)^{1/2}$ the right hand side stays bounded. 
        \paragraph{(II)} By definition it follows that, almost surely,
        \[
        \lim_{a \to a_0} |S(t,a)| = \lim_{a \to a_0} \left| \tfrac{1}{a - a_0} \left(\frac{ d\tilde{\bP} }{d\bP}(t,a) - 1 \right) -  m^{a_0}(t)\right| =  0 .
        \]
         Then, if we ensure uniform integrability, namely
         \begin{equation}
             \label{eq:uniform_integrability}
            \lim_{c\to 0} \sup_{a \in \Ieps} \E \left[ |S(t,a)|^2\mathbbm{1}_{|S(t,a)|^2\geq c}\right] = 0,
         \end{equation}
        then $|S(t,a)|^2$ converges to zero in expectation for $a \to a_0$, as desired.
         In particular, as
         \begin{equation*}
              \E \left[ |S(t,a)|^2\mathbbm{1}_{|S(t,a)|^2\geq c}\right] \leq \frac{1}{c}\,\E |S(t, a)|^4,
         \end{equation*}
         for \eqref{eq:uniform_integrability} to hold, it is sufficient to show that 
         \(
             \sup_{a\in \Ieps} \E |S(t,a)|^4 < \infty.
         \)
         By the definition \eqref{eq: def S(h)} of $S(t,a)$ and the triangular inequality we have 
         \begin{equation}
         \label{eq1S}
              \E |S(t,a)|^4 \leq 8\, \E \left| \frac{1}{a- a_0}\left(\frac{ d\tilde{\bP} }{d\bP}(t,a) - 1 \right) \right|^4 + 8\, \E |m^{a_0}(t)|^4.
         \end{equation}
         As $m^{a_0}(t)$ is Gaussian with zero mean and variance $t|D_a Q^{-1/2}f(a_0)|^2$, it follows that 
         \begin{equation*}
             \E |m^{a_0}(t)|^4 = 3 t^2 |D_a Q^{-1/2}f(a_0)|^4
         \end{equation*}
         which stays finite by \eqref{eq:conditionDf}. So we are left to show that 
         \begin{equation*}
               \E \left| \frac{1}{a- a_0}\left(\frac{ d\tilde{\bP} }{d\bP}(t,a) - 1 \right) \right|^4 
         \end{equation*}
         is uniformly bounded in $\Ieps$. By the mean value theorem in integral form and the definition of the density \eqref{eq: density girsanov} we have
          \begin{align}
              \E \left| \frac{1}{a- a_0}\left(\frac{ d\tilde{\bP} }{d\bP}(t,a) - 1 \right) \right|^4 = \E \left|\int_0^1 \exp(rM^a(t) - \tfrac{r}{2} \langle M^a\rangle_t )\frac{M^a(t) - \tfrac{1}{2}\langle M^a\rangle_t}{a - a_0}\, dr\right|^4 \nonumber\\
                \leq \E  e^{4|M^a(t)|}\left| \frac{M^a(t) - \tfrac{1}{2}\langle M^a\rangle_t}{a - a_0}\right|^4 \leq \left( \E  e^{8|M^a(t)|}\right)^{1/2} \left( \E \left| \frac{M^a(t)- \tfrac{1}{2}\langle M^a\rangle_t}{a - a_0}\right|^8 \right)^{1/2} \label{eq2}
         \end{align}
         where we have also used Cauchy-Schwartz inequality.
        %
        %
        As $M^a$ has Gaussian distribution with mean zero and variance $t|Q^{-1/2}(f(a) - f(a_0))|^2$, it can be shown that the right hand side of \eqref{eq2} is uniformly bounded by the mean value theorem applied to the variance, and by condition \eqref{eq:conditionDf}.

        We have proved that in \eqref{eq:longone} part (I) stays bounded and part (II) converges to zero. Then Equation~\eqref{eq:longone} ensures that $a \mapsto \cP_t^a\psi$ is differentiable in $a_0$ and its derivative at $a_0$ is 
        \begin{equation*}
            \left. \left( D_a \cP^a \psi \right) \right|_{a = a_0}(x) =  \E\,  \left [\,\psi(X_{a_0}(t, x))  m^{a_0}(t) \, \right].
        \end{equation*}
        
        We are left to show that \eqref{eq:ch5_bound_derivativeII} holds, namely the derivative is bounded as an operator from $\Cd$ to $\CU$, i.e.~there exists $C$ such that 
        \begin{equation*}
            \|\left. \left( D_a \cP^a \psi \right) \right|_{a = a_0}\|_U = \sup_{x} \frac{\left| \, \E\, \left[ \psi(X_{a_0}(t, x)) m^{a_0}(t) \right]\, \right|}{U(x)} \leq C \| \psi \|_{\td}
        \end{equation*}
        for all $\psi \in \Cd$. 
        Again, since $\E\,  m^{a_0}(t) = 0$ we have
        \begin{align*}
            \left|\,\E\, \left[\, \psi(X_{a_0}(t, x))   m^{a_0}(t)\,\right]\,\right|  &\leq \E\, \left|\left(\psi(X_{a_0}(t, x)) - \psi(x)\right)  m^{a_0}(t)\right|\\
             &\leq \,  \E\, \left[ |\psi(X_{a_0}(t, x))- \psi(x)|\left| \int_0^t Q^{-1} D_af(a_0)\, dW_s\right| \right].
        \end{align*}
        Then by \eqref{eq535}, Cauchy-Schwartz inequality and the It\^o isometry we have 
        \begin{align*}
        \leq  \| \psi\|_{\td}  \, \big( 1 + V(x)+ \E\,  V(X_{a_0}(t,x))\big)^{1/2}  |Q^{-1 /2}D_af(a_0)| 
        \end{align*}
        and, by the estimate \eqref{eq:lyapunovlinear}
        \begin{equation*}
            \leq  \| \psi\|_{\td}  \, \left(1 + K_{a_0} + V(x)(1 +C_{a_0}e^{-t\gamma_{a_0}}) \right)^{1/2}  |Q^{-1/2}D_af(a_0)| .
        \end{equation*}
        As $U \geq \left(1+V\right)^{1/2}$, setting
        \begin{equation*}
        C := |Q^{-1/2 }D_af(a_0)| \, t^{1/2}  \sup_{x\in \cH} \frac{\left(1 + K_{a_0} + V(x)(1 +e^{-t\gamma_{a_0}}) \right)^{1/2}}{U(x)} < \infty,
        \end{equation*}
        we have 
        \begin{equation*}
        \|\left. \left( D_a \cP_t^a \psi \right) \right|_{a = a_0}\|_U \leq C \| \psi\|_{\td}  .
        \end{equation*}
    \end{proof}
    
    We are now ready to show the following result: 
    
  \begin{theorem}
    \label{thm:responseU}
      Set $\Ieps = (a_0 - \varepsilon, a_0 + \varepsilon) \subset \R$ and consider the system \eqref{eq:ch5generalsde} with the map $a\mapsto f(a)$ differentiable as a map from $\Ieps$ into $\range Q$ and with $|Q^{-1/2}D_a f(a)|$ uniformly bounded, and let $\mu_a$ be the associated unique invariant measure. Suppose \nameref{asuA} holds for $a = a_0$, with Lyapunov function $V$. 
      Then
      the map $a \mapsto \langle \varphi, \mu_{a}\rangle $ is differentiable at $a = a_0$ for every $\varphi \in \Cd$ and the following identity holds
        \begin{equation*}
             \left. \dv{}{a} \langle \varphi, \mu_a \rangle\right|_{a = a_0} = \langle D_a \cP_t^{a_0} (1 - \cP_t^{a_0})^{-1}(\varphi - \langle \varphi, \mu_{a_0} \rangle), \mu_{a_0}\rangle. 
        \end{equation*}
    \end{theorem}
    \begin{proof}
        In order to apply \autoref{thm:general_response_HM} we show that \nameref{asuL} holds. Since \nameref{asuA} holds for $a = a_0$, then \autoref{thm:expconvergence} ensures $\cP_t^{a_0}$ exhibits a spectral gap, namely Assumption~\ref{L1}. 
        Next, \autoref{thm:diffsemigroup} with $U = \sqrt{1 + V}$ implies that the map $a \mapsto \cP_t^a \varphi$ satisfies Assumption~\ref{L2}. 
       Therefore we only have to ensure that Assumption~\ref{L3} holds for the choice of $U$.  
       As $V$ is a Lyapunov function for any $a \in \Ieps$, by definition of Lyapunov function and \eqref{eq:ourLyapunov} we have that, as we saw in \eqref{eq:ch4Vintegrable}, 
        \begin{equation*}
            \langle V, \mu_a \rangle \leq \frac{ K_\varepsilon }{\gamma_\varepsilon(1 - e^{-\gamma_\varepsilon t})} \fa t>0,
        \end{equation*}
        and as $\langle \sqrt{1 + V }, \, \mu_a \rangle\leq \sqrt{\langle 1 + V, \, \mu_a \rangle }$ Assumption~\ref{L3} follows.
    \end{proof}
    
   Assumption \ref{L2} and \ref{L3} can be shown for other choices of the semimetric $d$ in the definition of $\td$. In fact \autoref{thm:diffsemigroup} does not rely on the explicit definition \eqref{eq:ch4defdN} of $d_N$, but only on the fact that it is not larger than one. In the proof of \autoref{thm:responseU} we saw $U = \sqrt{1 + V}$ satisfies \ref{L3} thanks solely to the properties of the Lyapunov function. We introduced the semidistance $d_N$ in order to obtain \nameref{asuA} which provides quite general yet \emph{verifiable} conditions for SPDEs like \eqref{eq:ch5generalsde} to have a spectral gap.
    In \autoref{sec:response_NS} and \autoref{sec:response_QG} we will give two examples of application of \autoref{thm:responseU}, namely for the stochastic 2D Navier-Stokes equation with additive noise and the stochastic two--layer quasi--geostrophic model with additive noise on one of the layers. First though we close this section by studying when a SPDE like \eqref{eq:ch5generalsde} exhibits fractional response. 

\subsection{Fractional response}
    \label{subsec:Holder}
   So far we showed that, as a function of the parameter $a$, the invariant measure $\mu_a$ is weakly differentiable for observables in the space $\Cd$ when $f(a)$ is in the range of the noise.
    Under no restrictions on the spatial regularity of the forcing, we will still be able to show that $a\mapsto\mu_a$ is H\"{o}lder continuous as a map from $\Ieps$ into the space of functionals on $\Cd$, namely 
    there is $c = c(\varepsilon)$ such that for all $a_1, a_2\in \Ieps$
    \begin{equation}
    \label{eq:holderweak4}
       |\langle \varphi, \mu_{a_1} - \mu_{a_2}\rangle| \leq c \|\varphi \|_{\td}|a_1- a_2|^\alpha
    \end{equation}
    for an appropriate range of $\alpha\in (0, \alpha_0)$. 
    In order to prove \eqref{eq:holderweak4} we want to show that the conditions of \autoref{thm:generalmuHolder} hold. 
    Here we will provide a set of verifiable assumptions for SPDEs like \eqref{eq:ch5generalsde} to show \nameref{asuF}.
%
    \paragraph{Assumption~H} \label{asuH} 
    Let $\Ieps = ( a_0- \varepsilon, a_0+ \varepsilon )\subset \R$ be an interval and {consider \eqref{eq:ch5generalsde} with $a\mapsto f(a)$ being a $\beta$-H\"{o}lder continuous map from $\Ieps$ into $\cV'$. Furthermore} the following conditions hold:
    \begin{itemize}
         \item[{\crtcrossreflabel{H1}[H1]}] \nameref{asuA} holds for any $a\in \Ieps$;
         \item[{\crtcrossreflabel{H2}[H2]}]  Given $X_{a_1}, X_{a_2}$, solutions of \eqref{eq:ch5generalsde} for the same realisation of the noise and $a_1, a_2\in \Ieps$, there exists for all $t\geq 0$ a positive constant $C$ such that 
        \begin{equation*}
            |X_{a_1}(t)  - X_{a_2}(t)|^2 \leq C|a_1- a_2|^{2\beta} \exp(\kappa_1 \int_0^t \| X_{a_1}(s)\|^2 \, ds),
        \end{equation*}
        where $\kappa_1$ is as in \nameref{asuA}.
         \item[{\crtcrossreflabel{H3}[H3]}] There exists $c= c(a)$ with $\sup_{a\in \Ieps} c(a)< \infty$, $\chi>0$ such that 
          \begin{equation*}
              \E \, \exp(\alpha\upsilon |X_a(t,x)|^2) \leq c(a) \exp(\alpha\upsilon |x|^2 e^{-\chi t}).
          \end{equation*}
        %
            %
        %
    \end{itemize}
    Intuitively it is plausible that \ref{H2}, combined with a bound in $L^2(0,t; \cV)$ of $X_{a}$, implies H\"{o}lder continuity of $X_a$, and consequently of $\cP_t^a\psi$, for $\psi$ regular enough. %
    Less clear is the requirement of \ref{H3}. We will see in \autoref{thm:SPDEHolder} below that this bound implies condition \ref{F3}, thanks to the following lemma:
    \begin{lemma}
    \label{lemma:intexpmu_a}
       Let $X_a(t,x)$ be the solution of \eqref{eq:ch5generalsde} and suppose \nameref{asuA} holds. Then \nameref{H3} implies 
          \begin{equation}
          \label{eq:lemmaresult}
             \int \exp(\alpha\upsilon |x|^2)\, \mu_a(dx) < c(a).
          \end{equation}
    \end{lemma}
    
    %
    We postpone the proof of this technical lemma to the end of this section and now see how to apply it, together with \nameref{asuH} to show \nameref{F1}, \nameref{F2'} and \nameref{F3}. 
    \begin{theorem}
    \label{thm:SPDEHolder}
        Suppose \nameref{asuH} holds. Then for all $\alpha\in (0, \alpha_0)$, with $\alpha_0$ as in \eqref{eq:QGupsilonalpha0}, there exists $c = c(\varepsilon)>0$ such that 
        \begin{equation*}
             |\langle \varphi, \mu_{a_1} - \mu_{a_2} \rangle | \leq c\|\varphi\|_{\td} |a_1 - a_2|^\alpha \fa \varphi \in \Cd
        \end{equation*}
        for all $a_1, a_2\in \Ieps$, i.e.\ $\mu_a$ is locally $\alpha$-H\"{o}lder continuous with respect to the parameter $a$. 
     \end{theorem}
    
    \begin{proof}
   We want to apply \autoref{thm:generalmuHolder}. Thanks to \nameref{asuH} and \autoref{thm:expconvergence} we know $\cP_t^{a_1}$ exhibits a spectral gap and hence \ref{F1} holds.
   
     Next we show \ref{F2'}. By definition of the Wasserstein semidistance $W_{\td}$ we have 
        \begin{equation}
        \label{eq:ch5W<Ed}
            W_{\td}(P^{a_1}_t(x, \cdot), P^{a_2}_t(x, \cdot)) \leq \E \, \td(X_{a_1}(t), X_{a_2}(t))
        \end{equation} 
        and, thanks to the definition of $\td$ and the Cauchy-Schwartz inequality,
        \begin{equation}
            \E \, \td(X_{a_1}(t), X_{a_2}(t))
            \leq \sqrt{\E d_N(X_{a_1}(t), X_{a_2}(t))} \sqrt{1 + \E\, V(X_{a_1}(t)) + \E\, V(X_{a_2}(t)) }\label{eq:ch5general_EtildedafterCS}.
        \end{equation}
    Let us first bound $\E\, d_N(X_{a_1}(t), X_{a_2}(t))$. By definition of $d_N$ \eqref{eq:ch4defdN} and $\theta_\alpha$ \eqref{eq:ch4 theta(u,v)}
        \begin{equation*}
            d_N(X_{a_1}(t), X_{a_2}(t)) \leq  N|X_{a_1}(t) - X_{a_2}(t)|^{2\alpha} e^{\alpha\upsilon |X_{a_1}(t)|^2}
        \end{equation*}
    and, thanks to \ref{H2},
     \begin{equation}
     \label{eq:boundholderdN}
            d_N(X_{a_1}(t), X_{a_2}(t)) \leq  N C^\alpha |a_1- a_2|^{2\alpha \beta} \exp(\alpha\upsilon |X_{a_1}(t)|^2 + \alpha\kappa_1 \int_0^t \| X_{a_1}(s)\|^2 \, ds ).
        \end{equation}
    From \nameref{asuA}, part \ref{ref:A2} gives
        \begin{equation*}
            |X_{a_1}(t)|^2 + \kappa_2\int_0^t  \| X_{a_1}(s) \|^2 \, ds \leq  |x|^2 + \kappa_{\varepsilon} t + \Xi^{a_1}_\gamma, \quad t\geq 0.
        \end{equation*}
        Then, since $\upsilon = \kappa_1 /\kappa_2$, from \eqref{eq:boundholderdN} it follows 
       \begin{equation*}
          \E\, d_N(X_{a_1}(t), X_{a_2}(t)) \leq N C^\alpha |a_1 -a_2|^{2\alpha \beta} \exp(\alpha \upsilon |x|^2 + \alpha \upsilon \kappa_{\varepsilon} t) \E\, \exp(\alpha \upsilon \Xi^{a_1}_\gamma).
        \end{equation*}
        Thanks to the bound \eqref{eq:A2 martingale estimate} for $\Xi_\gamma^{a_1}$ it can be shown that
        \begin{equation*}
            \E\, \exp(\alpha \upsilon \Xi^{a_1}_\gamma)  \leq \frac{2\gamma}{2\gamma - \alpha \upsilon} =: C_\Xi
        \end{equation*}
        which is well defined since we consider $\alpha\in (0, \alpha_0)$ where $\alpha_0$ is as in \eqref{eq:QGupsilonalpha0}, i.e.
        \[ 
            \alpha_0 = \frac{1}{2}\wedge \frac{2\gamma}{\upsilon + 2\gamma }< \frac{2\gamma}{\upsilon}.
        \]
%
        Looking back at \eqref{eq:ch5general_EtildedafterCS} we have found that 
        \begin{equation*}
             \E \, \td(X_{a_1}(t), X_{a_2}(t)) \leq C_N |a_1 - a_2|^{\alpha \beta} e^{ \alpha\upsilon(|x|^2+ \kappa_3 t)/2 } \sqrt{1 + \E V(X_{a_1}(t)) + \E V(X_{a_2}(t))}.
        \end{equation*}
        Next, thanks to \ref{H1}, namely \ref{ref:A4}, for all $a\in \Ieps$, it holds 
        \begin{equation*}
           \E V(X_{a_1}(t)) + \E V(X_{a_2}(t)) \leq 2e^{-\gamma_\varepsilon t}V(x) + \tfrac{2K_\varepsilon}{\gamma_\varepsilon} \leq 2 V(x)+ \tfrac{2K_\varepsilon}{\gamma_\varepsilon}.
        \end{equation*}
        Then by \eqref{eq:ch5W<Ed} we have showed 
        \begin{equation*}
            W_{\td}(P_t^{a_1}(x, \cdot), P_t^{a_2}(x, \cdot))\leq \E \, \td(X_{a_1}(t), X_{a_2}(t)) \leq |a_1 - a_2|^{\alpha \beta}C(t)U(x)
        \end{equation*}
        with
        \begin{equation*}
            C(t) = C_N e^{ \alpha \upsilon \kappa_\varepsilon t/2 }\quand U(x) = e^{ \alpha \upsilon |x|^2/2}\left( 1 + \tfrac{2K_\varepsilon}{\gamma_\varepsilon} + 2V(x)\right)^{1/2}.
        \end{equation*}
        Last, we have to show that \ref{F3} holds, namely that
        \begin{equation*}
           \langle U, \mu_a \rangle = \int e^{ \alpha \upsilon |x|^2/2}\left( 1 + \tfrac{2K_\varepsilon}{\gamma_\varepsilon} + 2V(x)\right)^{1/2} \, \mu_{a}(dx) < \infty
        \end{equation*}
        uniformly in $a\in \Ieps$. 
        By Cauchy-Schwartz inequality we have 
        \begin{equation*}
           \langle U, \mu_a \rangle \leq 
          \left(  1 + \tfrac{2K_\varepsilon}{\gamma_\varepsilon} + 2\langle V, \mu_{a}\rangle \right) \int e^{ \alpha \upsilon |x|^2}  \mu_{a}(dx) .
        \end{equation*}
        The Lyapunov function is integrable against the invariant measure and, by \eqref{eq:ch4Vintegrable} and \eqref{eq:ourLyapunov} we have for all $a\in \Ieps$
        \begin{equation*}
            \langle V, \mu_a\rangle \leq \frac{K_\varepsilon}{\gamma_\varepsilon(1 - e^{-\gamma_\varepsilon t})}.
        \end{equation*}
        Finally, thanks to \ref{H3} and \autoref{lemma:intexpmu_a}, it follows that 
        \[ 
            \sup_{a\in \Ieps} \int e^{ \alpha \upsilon |x|^2}  \mu_{a}(dx) < \infty
        \]
        and \ref{F3} holds.
        
    \end{proof}
     We close the section by proving \autoref{lemma:intexpmu_a}. 
     \begin{proof}[Proof of \autoref{lemma:intexpmu_a}]
        Define the function \( \varphi(x) := \exp(\eta|x|^2) \) with $\eta := \alpha \upsilon$, and introduce an increasing sequence of cut-off functions $\chi_n\in [0, 1]$, i.e.~smooth functions supported on $[-n , n ]$ with $\chi_n = 1$ over $[-n+1, n -1]$ and $\chi_n \to 1 $ for $n \to \infty$. Then define the series of functions 
        \[
            \varphi_{n} (x) := \chi_n(|x|^2)\varphi(x), \quad n \in \N
        \]
        so that $\lim_{n\to \infty} \varphi_n = \varphi$.
        By the monotone convergence theorem we have 
        \begin{equation}
            \label{eq:monotoneconvergence}
           \lim_{n\to \infty} \, \langle \varphi_n , \mu_a \rangle = \langle \lim_{n\to \infty}\, \varphi_n , \mu_a \rangle = \langle \varphi , \mu_a \rangle,
        \end{equation}
        therefore we want to show that 
        \begin{equation*}
           \lim_{n\to \infty} \, \langle \varphi_n , \mu_a \rangle < \infty
        \end{equation*}
        uniformly in $a$. 
        We can write for any $y\in \cH$ and any $s_n>0$ 
        \begin{align*}
          \langle \varphi_n, \mu_a\rangle &=  \langle \varphi_n, \mu_a\rangle - \langle \varphi_{ n  ,t}, P^a_{s_n} (y, \cdot)\rangle + \langle \varphi_n, P^a_{ s_n } (y, \cdot)\rangle.
        \end{align*}
        We will show at the end of the proof that $\varphi_n$ are such that for some $C_1>0$
        \begin{equation}
        \label{eq:boundvarphi_nt}
             \| \varphi_n \|_{\td} \leq C_1 n^{1/2}\exp(\eta n ).
        \end{equation}
        %
                %
        Then from relation \eqref{eq: sup <f, mu> leq W_d}, since $\| \varphi_n \|_{\td}< \infty$ we have the following bound
        \begin{equation}
        \label{eq:NSholderproof4}
             \langle \varphi_n, \mu_a\rangle \leq \| \varphi_n \|_{\td}\, W_{\td}(\mu_a, P^a_{s_n } (y, \cdot)) + \langle \varphi_n, P^a_{ s_n } (y, \cdot)\rangle.
        \end{equation}
        Thanks to \nameref{asuA}, one has that \autoref{thm:expconvergence} holds and in particular there exists $r, C, t_0>0$ such that  
        \begin{equation*}
            W_{\td}(\mu_a, P^a_{s_n}(y, \cdot)) \leq C(1 + V(y))e^{-rs_n} \fa s_n \geq t_0.
        \end{equation*}
       Further, using \eqref{eq:boundvarphi_nt} and adjusting appropriately the constant $C$, from \eqref{eq:NSholderproof4} we have 
        \begin{equation}
             \label{eq:proofU_2}
             \langle \varphi_n, \mu_a\rangle \leq C n^{1/2} \exp(\eta n - r s_n )(1+ V(y)) + \langle \varphi_n, P^a_{s_n} (y, \cdot)\rangle.
        \end{equation}
       The first term converges to zero if we choose $s_n := 2\eta n/r \vee t_0$.
        For the second term on the right hand side of \eqref{eq:proofU_2}, the definition of $\varphi_n$ and \ref{H3} give
        \begin{align*}
            \langle \varphi_n, P^a_{s_n} (y, \cdot)\rangle
            &\leq \E\, \exp(\eta|X_a(s_n,y)|^2) \leq c(a)\exp(\eta e^{- \chi s_n} |y|^2).
        \end{align*}
        which converges to $c(a)$.
        It follows immediately that 
        \begin{equation}
        \label{eq:proofU_3}
            \lim_{n \to \infty} \langle \varphi_n, \mu_a\rangle \leq c(a)
        \end{equation} where, recall $c(a)$ is assumed to be uniformly bounded in $\Ieps$.
        In summary, given \eqref{eq:monotoneconvergence} and \eqref{eq:proofU_3}, one has the desired result \eqref{eq:lemmaresult}.
     
        We conclude the proof by showing that the estimate \eqref{eq:boundvarphi_nt} for $\| \varphi_n\|_{\td}$ holds.
        By the mean value theorem, given $z\in [x,y]$,
         \begin{align*}
            |\varphi_n(x) - \varphi_n(y)| &\leq \| D_x \varphi_n(z)\| |x-y| 
        \end{align*}
        and so if $\| D_x \varphi_n(z)\|$ is bounded uniformly in $z$, we have the following bound
        \begin{equation}
        \label{eq:varphint}
             \| \varphi_n \|_{\td} \leq \left(\sup_{z\in H}\| D_x \varphi_n(z)\|\right)\left( \sup_{x\neq y} \frac{|x-y|}{\td(x,y)}\right).
        \end{equation}
        Focusing on the derivative of the functions $\varphi_n$ with respect to $x$ it is easy to see that
        \begin{equation*}
           \| D_x \varphi_n(z)\| \leq 2|z| \varphi(z)  \left|\chi_n'(|z|^2) + \chi_n(|z|^2) \eta \right|.
        \end{equation*}
        The smooth cut--off function 
        \[
        \chi_n(x)= \left\{ \begin{array}{ll} 1 & \mbox{for } x < n-1 \\
        \chi_1(x-n+1) & \mbox{for } n-1 \leq x \leq n \\
        0 & \mbox{for } x > n\end{array} \right.
        \]
       hence its derivative $\chi_n'(z)$ is well defined for any choice of $n$ and has support $[n-1, n]$ and is uniformly bounded in $n$. Therefore
        \begin{align*}
            \sup_{z\in \cH} \| D_x \varphi_n(z)\| &\leq \sup_{|z|^2 \leq n} 2|z| \exp(\eta |z|^2 ) |  \chi_n'(|z|^2) + \eta \chi_n(|z|^2) | \\
            &\leq  2n^{1/2} \exp(\eta n ) \left( \sup_{z\in \cH} | \chi_n'(|z|^2)|  + \eta \right) . 
        \end{align*}
        %
        %
        Therefore we showed that there exists a positive constant 
        \[ 
            C_1 := 2 \left( \sup_{t \in [0, 1]} | \chi_1'(t)|+ \eta\right)
        \]
        such that the derivative of $\varphi_n$ satisfies
        \begin{equation*}
             \sup_{z\in \cH} \| D_x \varphi_n(z)\| \leq C_1 n^{1/2} \exp(\eta n).
        \end{equation*}
        Finally from \eqref{eq:varphint} we see
        \begin{equation*}
             \| \varphi_n \|_{\td} \leq C_1 n^{1/2} \exp(\eta n )  \sup_{x\neq y} \frac{|x-y|}{\td(x,y)}.
        \end{equation*}
        By the definition of the semimetric $\td$ in \eqref{eq:ch4defdN} we have
        \begin{equation*}
            \sup_{x\neq y} \frac{|x-y|}{\td(x,y)} < \infty
        \end{equation*}
        so that, relabelling $C_1$ appropriately, the desired result holds.
    \end{proof}
    
%
\section{Stochastic Navier--Stokes equations}
\label{sec:response_NS}
Let $\cD= [0,L]\times [0,L]\subset \R^2$ with $L>0$  and consider the two--dimensional (2D) stochastic Navier--Stokes equation on $\cD$ 
\begin{equation}
    \label{eq:SNS_strong}
    \begin{split}
         &d u + (u \cdot \nabla) u \; dt =  \left(\nu \Delta u  - \nabla p + f(a)\right) \, dt + dW  \\
       & \operatorname{div} u = 0 \\
        &u(0,x) = u_0 .
    \end{split}
\end{equation}
 Here $u = u(t,x)$ is the velocity of an incompressible fluid, $\nu$ is the viscosity,  $p(t,x)$ the pressure of the fluid, $f(a)$ is a time-independent deterministic forcing depending on a parameter $a\in \R$, and $W$ is a $Q$-Wiener process on a probability space $(\Omega, \mathcal{F}, \mathbb{P})$ with the covariance operator $Q$. We consider \eqref{eq:SNS_strong} with periodic boundary conditions and we assume that the average flow vanishes, namely 
\[ \int_\cD u(t,x) \, dx = 0  \fa t \geq 0.\]

In this section, after setting the necessary notations, we discuss a spectral gap results available in the literature for this model and apply the methodology developed in \autoref{subsec:linear} and \autoref{subsec:Holder} to establish linear and fractional response. 

\subsection{Mathematical set up}
Let $L^2(\cD)$ and $H^k(\cD)$, $k \in \N$, be the Sobolev spaces of $L$-periodic functions such that 
\[ \int \varphi(x) \, dx = 0\]
and $H^{-k}$, $k\in \N$ the dual space of $H^k$. 
%
As the velocity $u = u(t,x)$ is two--dimensional, it is natural to introduce the following product spaces
\begin{equation*}
    \bLtwo(\cD)= \left[ L^2(\cD)\right]^2 \quad \text{and} \quad \bH^k(\cD)= \left[ H^k(\cD)\right]^2 
\end{equation*}
and we consider the Hilbert spaces 
\begin{align*}
\begin{split}
       \cH &= \lbrace u\in \mathbf{L}^2(\cD) \, : \, \operatorname{div} u = 0 \, \text{ in } \, \cD \rbrace \\
        \cV &= \lbrace u\in \mathbf{H}^1(\cD) \, : \, \operatorname{div} u = 0 \, \text{ in } \, \cD \rbrace
\end{split}
\end{align*}
with norms $|\cdot|$ and $\| \cdot\| $. Elements of $\cH$ and $\cV$ then satisfy the divergence free condition and the boundary conditions by definition.

Let $A$ denote the Stokes operator which we consider as an operator on $\cH$; it can be shown that $\cV = D(A^{1/2})$ and $\| u \| = |A^{1/2}u|$. 
 Since we consider periodic boundary conditions, we have that $A u = - \Delta u $ for all $u\in D(A)$. Moreover, the operator $A$ is a self-adjoint positive operator on $\cH$, and we denote by $\lbrace \lambda_k \rbrace$ its eigenvalues and by $\lbrace e_k \rbrace$ a corresponding complete orthonormal system of eigenvectors. 

Denote by $\cV'$ the dual of $\cV$, then we have 
\begin{equation*}
    D(A) \subset \cV \subset \cH \subset \cV',
\end{equation*}
where the inclusions are continuous and each space is dense in the following one. 
The covariance operator $Q$ is a nonnegative, symmetric and trace class operator in $L^2$. We assume that $Q$ and the Stokes operator $A$ commute.

Crucial part of the study of the Navier--Stokes equations is the treatment of the advection term
\begin{equation*}
     \langle B(u,v), w \rangle = b (u,v,w):= \int_\cD w(x) \cdot (u(x) \cdot \nabla) v(x) \, dx.
\end{equation*}
 It is easily seen that the trilinear form $b$ satisfies 
    \begin{equation*}
       b(u,v,v) = 0 \fa u\in \cH, \; v\in \cV
    \end{equation*}
and thanks to Ladyzhenskaya's inequalities there exists a positive constant $k_0$ so that
    \begin{equation}
    \label{eq:NS_lady}
        |b(u, v, u)|\leq k_0|u| \| u\| \|v\| \fa u,\,  v \in \cV.
    \end{equation}
By classical arguments (see e.g.~\cite{temam2001navier,Robinson}), the following weak formulation of \eqref{eq:SNS_strong} is obtained
\begin{equation}
\label{eq:SNS_weak}
    du + \left(\nu Au + B(u,u)\right)\, dt = f(a) \, dt + dW,\quad u(0,x) = u_0(x).
\end{equation}
%
%
Given $f(a) \in \cV'$ and $u_0\in \cH$, for any $T>0$ there exists a unique solution $u = u(t, \omega ; u_0, a)$ in $C([0, T]; \cH) \cap L^2([0,T],\cV)$ for almost all $\omega$ such that \eqref{eq:SNS_weak} holds in $\cV'$ and the associated Markov semigroup $\cP_t^a$ on $\cH$ is Feller (see e.g.~\cite{Flandoli94,Ferrario2001}).
%
\subsection{Spectral gap}
\label{subsec:NSspectral}
    From the literature it is known that this model exhibits exponential convergence of transition probabilities, and \cite{KulikSch18, butkovsky2020} in particular showed exponential convergence with respect to the Lipschitz seminorm $\| \cdot \|_{\td}$ with Lyapunov function $V(x) = |x|^2$, namely 
    \begin{equation*}
        \td(x,y) ^2 =\left( N |x - y|^{2\alpha}e^{\alpha \upsilon|x|^2}\wedge  N |x - y|^{2\alpha}e^{\alpha \upsilon|y|^2} \wedge 1\right)\left( 1 + |x|^2 + |y|^2\right).
    \end{equation*}
    More specifically, for a fixed parameter $a$ \nameref{asuA} has been shown to hold in \cite{butkovsky2020} (and in \cite[Section~4.2]{thesis} using the same framework and notation as in the present paper). There the chosen controlled equation is
    \begin{equation*}
        d\tu + \left(\nu A\tu + B(\tu,\tu)\right)\, dt = \left(f(a) + \tfrac{\nu \lambda_n}{2}\Pi_n \left(u - \tu\right)\right)\, dt + dW, \quad \tu(0,x) = \tu_0(x)\neq u_0(x)
    \end{equation*}
    where $\Pi_n$ is the projection onto the first $n$ eigenfunctions of the Stokes operator for an appropriate $n\in \N$ to be taken large enough. As the dependence on the forcing $f(a)$ of the estimates is explicit there, one can easily create a uniform version of the estimates, hence we do not repeat the full argument here.
        %
    %
    Table~\ref{tab:parametertranslate} provides a summary of the parameters mentioned in \nameref{asuA} together with expressions in terms of parameters appearing the stochastic Navier--Stokes equation which will be used in the results of this section.
    %
    %
    \begin{table}[ht]
    \begin{center}
    \begin{tabular}{ll}
    \nameref{asuA} & Navier--Stokes \\[1ex] 
    $\kappa_0$ &  $\nu\lambda_n$\\ 
    $\kappa_1$ & $k_0^2/\nu$ \\ 
    $\kappa_2$ & $\nu- \gamma \lambda_1^{-1}\Tr Q $  \\ 
    $\kappa_\varepsilon$ & $\Tr Q + \sup_{|a - a_0|< \varepsilon}\| f(a)\|_{-1}^2/\nu$  \\ 
      $\gamma_{\varepsilon}$ & $\nu \lambda_1$  \\ 
          $K_\varepsilon$ & $\Tr Q +\sup_{|a - a_0|< \varepsilon} \| f(a)\|_{-1}^2/\nu $  
    \end{tabular} 
    \end{center}
    \caption{\label{tab:parametertranslate}
    Parameters from \nameref{asuA} and their values in terms of parameters from the Navier--Stokes equations.
    Here $\lambda_n$ is the $n$-th eigenvalue of $-\Delta$, $k_0$ appears in the estimates of the trilinear form~\eqref{eq:NS_lady}, $\nu$ is the viscosity, $\gamma$ is chosen in such a way that $\kappa_2$ is positive, $Q$ is the covariance operator of the noise and $f$ is the deterministic external forcing.
}
    \end{table}
    %
    %
    \subsection{Linear and fractional response}
    \label{subsec:NSHolder}
   Given the general framework obtained in \autoref{thm:responseU}, we have the following result on linear response for the 2D stochastic Navier--Stokes equations:
    \begin{theorem}
    \label{thm:NSresponse}
        Set $\Ieps = (a_0 - \varepsilon, a_0+ \varepsilon)\subset\R$, $a_0\in \R$, $\varepsilon>0$. Consider the Navier--Stokes equation \eqref{eq:SNS_weak} with deterministic forcing the map $a\mapsto f(a)$ continuously differentiable as a map from $\Ieps$ into $\range Q$ with $|D_a Q^{-1/2}f(a)|$ uniformly bounded in $a$. Let $\mu_a$ be the associated unique invariant measure.
        Then the map $a \mapsto\langle \varphi, \mu_a\rangle $ is differentiable at $a = a_0 $ for every $\varphi\in \Cd$ and the following identity holds
        \begin{equation*}
             \left. \dv{}{a} \langle \varphi, \mu_a \rangle\right|_{a = a_0} = \langle D_a \cP_t^{a_0} (1 - \cP_t^{a_0})^{-1}(\varphi - \langle \varphi, \mu_{a_0} \rangle), \mu_{a_0}\rangle. 
        \end{equation*}
    \end{theorem}
    \begin{proof}
    As discussed in \autoref{subsec:NSspectral}, \nameref{asuA} holds for the Navier-Stokes equations \eqref{eq:SNS_weak} with Lyapunov function $V(x) = |x|^2$. Then \autoref{thm:responseU} applies as the forcing $f(a)$ is assumed in the range $Q$.
    \end{proof}
    Next, we consider a deterministic forcing $f(a)$ being $\beta$-H\"{o}lder continuous in the parameter $a\in \R$ as a function into $\cV'$. By showing that conditions of \nameref{asuH} holds we will ensure the model admits fractional response.
    \begin{theorem}
    \label{thm:NSHolder}
         Let $\Ieps = (a_0- \varepsilon, a_0 + \varepsilon)$ be an arbitrary neighbourhood of $a_0\in \R$.
        Consider the Navier--Stokes equation \eqref{eq:SNS_weak} with 
        $f(a)$ locally $\beta$-H\"{o}lder continuous in $a$ as a function into $\cV'$, namely for every $\varepsilon>0$ there exists $C_f = C_f(a_0, \varepsilon)$ such that 
        \begin{equation}
        \label{eq:holder_f}
            \|f(a_1) - f(a_2)\|_{-1}\leq C_f|a_1 - a_2|^\beta \fa a_1, a_2 \in \Ieps.
        \end{equation}
        Then for all $\alpha \in (0, \alpha_0)$, with $\alpha_0$ as in \eqref{eq:QGupsilonalpha0}, for every $\varepsilon$ there exists $c = c(\varepsilon)$ such that for all $a_1, a_2 \in\Ieps$
        \begin{equation*}
            |\langle \varphi, \mu_{a_1} - \mu_{a_2}\rangle | \leq c\|\varphi\|_{\td}|a_1 - a_2|^{\alpha\beta} \fa \varphi\in \Cd.
        \end{equation*}
    \end{theorem}
    \begin{proof}
        Given \autoref{thm:SPDEHolder} we have to ensure that the solution of \eqref{eq:SNS_weak} satisfies \nameref{asuH}. As discussed in \autoref{subsec:NSspectral}, \nameref{asuA} holds for the Navier-Stokes equations \eqref{eq:SNS_weak} with Lyapunov function $V(x) = |x|^2$, and the parameters in \autoref{tab:parametertranslate} are bounded uniformly for $a\in \Ieps$, so \ref{H1} is satisfied.
        \paragraph{Proof of \ref{H3}} 
        Let $\eta>0$ and take the $\cH$ product of \eqref{eq:SNS_weak} with $\eta u$ itself to get 
        \begin{equation*}
            d(\eta |u(t)|^2) = \eta \left( 2 \langle f(a), u \rangle + \Tr Q  - 2\nu \|u\|^2 \right) \, dt + 2\eta \langle u, \cdot \rangle dW(t).
        \end{equation*}
    To this stochastic differential equation we apply the following result~\cite[Lemma~5.1]{HMatt08}
      \begin{lemma}
        \label{lemma: Lemma5.1HMatt08}
        Let $M$ be a real-valued semimartingale 
        \begin{equation*}
            dM(t, \omega) = F(t, \omega)\, dt + G(t, \omega)\, dB
        \end{equation*}
        where $B$ is a standard Brownian motion. Assume there exists a process $Z$ and positive constants $b_1, b_2, b_3$ with $b_2> b_3$, such that 
        \begin{enumerate}[label = (\roman*)]
            \item $F\leq b_1 - b_2 Z$ a.s.,
            \item  $M\leq Z$ a.s.,
            \item $G^2\leq b_3 Z$ a.s.
        \end{enumerate}
        Then the bound 
        \begin{equation*}
            \E \, \exp(M(t) + \frac{b_2 e^{-b_2t/4}}{4}\int_0^t Z(s) \, ds ) \leq \frac{b_2 \exp(\frac{2b_1}{b_2})}{b_2 - b_3} \exp(M(0)e^{-b_2t/2})
        \end{equation*}
        holds for any $t\geq 0$. 
    \end{lemma}

         We plan to apply \autoref{lemma: Lemma5.1HMatt08} with $M(t) := \eta |u_t|^2$ and $Z(t) := \eta \lambda_1^{-1} \|u_t\|^2$.
        We establish the three conditions of the lemma in turn: %
        \begin{enumerate}[label = (\roman*)]
            \item The first condition of the lemma is satisfied for $b_1 = \eta\left( \|f(a)\|_{-1}^2/\nu + \Tr Q  \right)$ and  $b_2 = \nu \lambda_1 $ as 
                \begin{align*}
                    F(t) &= \eta \left( 2  \langle f(a), u \rangle + \Tr Q  - 2\nu \|u\|^2 \right)\\
                    &\leq \eta \left( \frac{ \|f(a)\|_{-1}^2}{\nu} + \nu \|u\|^2 + \Tr Q  - 2\nu \|u\|^2 \right) = b_1 - b_2 Z(t).
                \end{align*}
            \item The second condition follows simply by Poincar\'{e}'s inequality $M(t) = \eta |u|^2 \leq \eta \lambda_1^{-1} \|u\|^2 = Z(t)$.
            \item For the third condition note that we can write
            \[
               2\eta\langle u, \cdot\rangle\, dW = 2\eta\|\langle u, \cdot \rangle\|_{L_2^0} \frac{\langle u, \cdot\rangle}{\|\langle u, \cdot \rangle\|_{L_2^0}}\, dW = 2\eta\|\langle u, \cdot \rangle\|_{L_2^0} \, dB
               \]
              where $B$ defined as above is a standard real-valued Brownian motion. Then
               \[
               G(t, \omega)^2 = 4\eta^2 \|\langle u, \cdot \rangle\|_{L_2^0}^2 = 4 \eta^2 \sum_{k\in \N} \left| \left( u, Q^{1/2}e_k\right)\right|^2 \leq 4\eta^2 \Tr Q  |u_t|^2 \leq 4\eta^2 \Tr Q  \lambda_1^{-1}\|u_t\|^2
            \]
            and $b_3 = 4\eta \Tr Q$.
        \end{enumerate}
        To ensure that $b_2> b_3$ i.e.~$4 \eta \Tr Q  < \nu \lambda_1 $ we take 
        \begin{equation}
        \label{eq:NSboundeta1}
              \eta < \nu \lambda_1 / 4\Tr Q =: \eta_1.
        \end{equation}
        Then \autoref{lemma: Lemma5.1HMatt08} gives
        \begin{equation*}
            \E \exp( \eta |u_t|^2 + \frac{\nu \eta  e^{- \nu \lambda_1 t /4}}{4}\int_0^t \|u_s\|^2 \, ds ) \leq 
            c(a) \exp(\eta|u_0|^2 e^{- \nu \lambda_1 t /2})
        \end{equation*}
        with 
        \begin{equation*}
            c(a) =  \frac{\nu \lambda_1 \exp(\frac{2\eta( \Tr Q +  \| f(a)\|_{-1}^2/\nu )}{\nu \lambda_1})}{\nu \lambda_1 - 4\eta \Tr Q }
        \end{equation*}
       which stays uniformly bounded for all $a\in \Ieps$. 
        Consequently we have
        \begin{equation*}
            \E \exp( \eta |u_t|^2 )\leq c(a) \exp(\eta|u_0|^2 e^{- \nu \lambda_1 t /2}) 
        \end{equation*}
        for all $\eta \in (0 , \eta_1)$. 
          Then we only have to make sure we can take $\eta = \alpha \upsilon$.
        First note that given the definition \eqref{eq:QGupsilonalpha0} of $\alpha_0$ we have 
        \begin{equation*}
            \alpha_0 \upsilon = \frac{\upsilon}{2}\wedge \frac{2\gamma \upsilon}{2\gamma + \upsilon} < 2\gamma
        \end{equation*}
        where $\gamma>0$ is an arbitrary parameter smaller than $\nu \lambda_1 /\Tr Q $.
        Therefore if we choose 
        \begin{equation*}
            0< 2\gamma< \eta_1 = \frac{\nu \lambda_1}{ 4 \Tr Q },
        \end{equation*}
         we have \( \alpha_0 \upsilon < \eta_1 \) as desired.
        
    \paragraph{Proof of \ref{H2}}
      Set $u(t) := u(t; u_0, a_1)$, $v(t):=u(t; u_0, a_2)$ and $  w := u - v $. Then $w$ must satisfy the following equation 
    \begin{equation}
    \label{eq:NS eq for w}
        \dfrac{dw}{dt} + \nu Aw + B(w,u) + B(v,w)=  f(a_1) - f(a_2), \quad w(0) = 0.
    \end{equation}
    Take the $\cH$ scalar product of \eqref{eq:NS eq for w} with $w$ 
    \begin{equation*}
        \frac{1}{2}\frac{d|w|^2}{dt} + \nu \|w\|^2 + (B(w,u),w) = \langle  f(a_1) - f(a_2), w\rangle, 
    \end{equation*}
    where we have used that $ (B(v,w), w) = 0$. Using the estimate \eqref{eq:NS_lady} for the trilinear form and Cauchy-Schwartz inequality, we get
    \begin{equation*}
        \frac{1}{2}\frac{d|w|^2}{dt} + \nu \|w\|^2 \leq k_0 |w| \|w\| \|u\| +  \| f(a_1) - f(a_2)\|_{-1}\|w\|,
    \end{equation*}
    and by Young's inequality and the H\"{o}lder continuity of $f$ \eqref{eq:holder_f}
    \begin{equation*}
        \frac{1}{2}\frac{d|w|^2}{dt} \leq \frac{k_0^2}{2\nu} \|u\|^2|w|^2 +  \frac{C_f^2}{\nu}| a_1 - a_2|^{2\beta}. 
    \end{equation*}
    By Gronwall's inequality we get
        \begin{equation*}
        \begin{split}
            |w(t)|^2 \leq \tfrac{2C_f^2}{\nu}|a_1 - a_2|^{2 \beta } \int_0^t\exp(\tfrac{k_0^2}{\nu} \int_s^t \| u(r,a_1) \|^2 \, dr) \, ds\\
            \leq  \tfrac{2C_f^2}{\nu}|a_1 - a_2|^{2\beta}\, t \exp(\tfrac{ k_0^2}{\nu} \int_0^t \| u(r,a_1) \|^2 \, dr)
        \end{split}
        \end{equation*}
        so that \ref{H2} is satisfied with $C =2 t C_f^2/\nu$, $\kappa_1 = k_0^2/ \nu$.
    \end{proof}
\section[Stochastic two--layer quasi--geostrophic model]{Stochastic two--layer quasi--geostrophic \\model}
\label{sec:response_QG}
 The 2LQG equations model mid-latitude atmosphere and ocean dynamics at large scale. The model describes two layers of fluid one on top of the other with mean height $h_1$ for the top layer and $h_2$ for the bottom one, and with density respectively $\rho_1$ and $\rho_2$ with $\rho_1< \rho_2$. 
We consider the so-called $\beta$-plane approximation to the Coriolis effect (see \cite[Section~2.3.2]{Vallis06}).
We assume that the forcing acts only on the top layer and has a non-trivial stochastic part 
which accounts for example for the effect of the wind on the upper ocean. For  a more detailed exposition of the mathematical description see \cite{carigi2022exponential} and references therein.

Let $\cD$ be a squared domain $\cD = [0,L]\times [0,L]\subset \R^2$. Consider the following equations 
\begin{align}
\label{eq:QG_stochastic}
\begin{split}
    &d q_1 + J(\psi_1, q_1 + \beta y ) \, dt = \left(\nu\Delta^2\psi_1 \,+  f(a) \right) dt + d W\\
    & \partial_t q_2 + J(\psi_2, q_2 + \beta y ) = \nu\Delta^2\psi_2 - r\Delta \psi_2,
\end{split}
\end{align}
 where $\x = (x, y)\in \cD$, $\bfpsi(t, \x) = (\psi_1(t, \x), \psi_2(t, \x))^t$ is the streamfunction of the fluid, and $\q(t,\x) = (q_1(t, \x), q_2(t, \x))^t $ is the so-called quasi--geostrophic potential vorticity. Vorticity and streamfunction are related through
 \begin{equation}
\label{eq:simple_relation_q_psi}
\begin{split}
     q_1 = \Delta \psi_1 + F_1(\psi_2 - \psi_1) \\
     q_2 = \Delta \psi_2 + F_2(\psi_1 - \psi_2),
\end{split}
\end{equation}
where $F_1, F_2$ are positive constants. 
Moreover, $J$ is the Jacobian operator $J(a,b) = \nabla^{\perp}a \cdot \nabla b$, $W$ is a Wiener process on a probability space $(\Omega, \mathcal{F}, \bP)$, with covariance operator $Q$. Furthermore we assume periodic boundary conditions for $\bfpsi$ in both directions with period $L$ and we impose that 
\begin{equation}
\label{eq:ch1zeromeanvalue}
    \int_\cD \bfpsi(t, \x) \, d\x = 0 \fa t\geq 0.
    \end{equation}
The model includes a deterministic forcing on the top layer $f(a)$ (time-independent) as well 
with zero spatial averages, i.e.
\begin{equation*}
    \int_\cD f(\x, a) \, d\x = 0 .
\end{equation*}
The constants $F_1, F_2$ are such that 
\begin{equation*}
    h_1 F_1 = h_2F_2 =: p.
\end{equation*} 
The model \eqref{eq:QG_stochastic} includes dissipation generated by the eddy viscosity on both layers modeled by the terms $\nu\Delta^2\psi_i$ and by the friction with the bottom modeled by $r\Delta \psi_2$.
We can write \eqref{eq:QG_stochastic} in vectorial formulation introducing 
\begin{equation}
    \label{eq: def B(U,V)}
  B(\bfpsi, \bfxi) = \left( \begin{array}{r}J( \psi_1, \Delta \xi_1) + F_1J(\psi_1, \xi_2) \\ J( \psi_2, \Delta \xi_2) + F_2J(\psi_2, \xi_1) \end{array} \right) .
\end{equation} 
Then, using the fact that $J(\psi,\psi) =0$, \eqref{eq:QG_stochastic} becomes
\begin{equation}
\label{eq:QG_stoc_vec}
    d \q + \left( B(\bfpsi, \bfpsi) + \beta \partial_x\bfpsi \right) \, dt= \nu \Delta^2\bfpsi \,dt+  \binom{f(a)}{- r \Delta \psi_2} \, dt+ d\mathbf{W}
\end{equation}
where $\mathbf{W} = (W, 0)^t$, and $\Delta \bfpsi = (\Delta \psi_1, \Delta \psi_2)^t$. Moreover, we can express the relation \eqref{eq:simple_relation_q_psi} between the streamfunctions and the vorticities as
\begin{equation*}
    \q = (\Delta + M )\bfpsi \quad \text{with } 
    M = \begin{pmatrix} 
    - F_1 & F_1 \\F_2 & - F_2
    \end{pmatrix}.
\end{equation*}
Next we set the notations for the mathematical setup of two--layer quasi--geostrophic model used in this work following \cite{carigi2022exponential}.

\subsection{Mathematical set up}
Let $(L^2(\cD), \|\cdot\|_0)$, $(H^k(\cD), \| \cdot \|_k)$, $k\in \R$ be the standard Sobolev spaces of $L$-periodic functions satisfying \eqref{eq:ch1zeromeanvalue}. Denote by $(\cdot , \cdot)_k$ the associated scalar product. 
We introduce the product spaces to deal with our coupled system
\begin{equation*}
    \bLtwo(\cD) = \left[ L^2(\cD) \right]^2 \quand \bH^k(\cD) = \left[ H^k(\cD)\right]^2
\end{equation*}
with the weighted scalar product and norm
\begin{align*}
    (\bfpsi, \bfxi)_k := h_1(\psi_1, \xi_1)_{k} + h_2(\psi_2, \xi_2)_{k} \nonumber\\
    \|\bfpsi\|_k^2 := h_1\|\psi_1\|_{k}^2 + h_2\|\psi_2\|_{k}^2 
\end{align*}
for $\bfpsi$ and $\bfxi$ elements of $H^k \times H^k$, $k>0$ or $L^2 \times L^2$ for $k = 0$.
Further denote with $\bH^{-k}$ the dual space of $\bH^k$, $k>0$. 

We take the covariance operator $Q$ to be nonnegative, symmetric and trace class in $L^2$. We also assume that $Q$ and $-\Delta$ commute.

Define the operator $\tA: \mathbf{H}^{k+2} \to \mathbf{H}^{k}$, $k\in \R$, connecting the streamfunction with the quasi--geostrophic potential vorticity
\begin{equation*}
    \tA\bfpsi = - (\Delta + M)\bfpsi, \quad \bfpsi \in \mathbf{H}^{k+2}.
\end{equation*}
It is easy to see that $\tA$ is an unbounded non--negative self--adjoint operator in $\bH^k$ with respect to the weighted scalar product $(\cdot , \cdot)_k$
and thanks to \eqref{eq:ch1zeromeanvalue}, $\tA$ has a bounded inverse which is bounded as function $\mathbf{H}^{k} \to \mathbf{H}^{k+2}$, that is, for each $\q\in \bH^k$ there exists a unique $\bfpsi\in \mathbf{H}^{k+2}$ such that $\q = - \tA \bfpsi$.
%
%
\begin{remark}
    Since the $\bLtwo$ and $\bH^1$ norms and the $\bLtwo$ scalar product are the most used throughout this section, for the sake of simplifying notation we denote them as follows
    \begin{align*}
        |\psi| := \| \psi\|_0 &\quand \| \psi \| := \|\psi\|_1 \\
        |\bfpsi| = h_1|\psi_1| + h_2|\psi_2| := \| \bfpsi\|_0 &\quand \| \bfpsi \| = h_1 \| \psi_1\| + h_2 \|\psi_2\| := \| \bfpsi\|_1\\
        (\psi,\xi) := (\psi,\xi)_0 &\quand (\bfpsi, \bfxi) = h_1(\psi_1, \xi_1) + h_2(\psi_2, \xi_2)
    \end{align*}
\end{remark}
Finally, we introduce two new norms on the level of the potential vorticities. 
For $\q\in \Hmo$ there exists $\bfpsi\in \bH^1$ such that $\q = -\tA\bfpsi$, and we can define the norm on $\Hmo$
\begin{equation*}
    \vertiii{\q}_{-1}^2 := \| \bfpsi\|^2 + p | \psi_1 - \psi_2|^2
\end{equation*}
and, for $\q\in \bLtwo$ with $\bfpsi \in \bH^2$ define the norm on $\bLtwo$
\begin{equation*}
    \vertiii{\q}_0^2 :=|\Delta \bfpsi|^2 + p \|\psi_1 - \psi_2\|^2.
\end{equation*}
Note that by Poincar\'{e} inequality one has
\begin{align*}
\begin{split}
   \vertiii{\q(t)}^2_{-1} &= \| \bfpsi\|^2 + p |\psi_1 - \psi_2|^2 \\
   &\leq \lambda_1^{-1}\left( |\Delta \bfpsi|^2 + p \| \psi_1 - \psi_2\|^2\right) = \lambda_1^{-1} \vertiii{\q(t)}_0^2.
  \end{split}
\end{align*}
Furthermore, these norms are equivalent to $\| \cdot \|_{-1}$ and $\|\cdot \|_{0}$ respectively and have a series of useful properties: 

\begin{lemma}[{\cite[Lemma~2.2]{carigi2022exponential}}]
\label{lemma:ch1propertiesnorms}
    Consider $\q\in \Hmo$ and $\bfpsi \in \mathbf{H}^1$ such that $\q = -\tA\bfpsi $. Then the following relations hold:
    \begin{align}
        -( \q, \bfpsi) = \vertiii{\q}_{-1}^2 \label{eq:ch1(u,v)=|u|*} \\
        \|\bfpsi \|^2 \leq \vertiii{\q}_{-1}^2 \leq c_0 \|\bfpsi \|^2 \label{eq:ch1*normpoincare}
    \end{align}
    for $c_0 = 1 + 2\lambda_1^{-1}\max(F_1, F_2)$.
    For $\q\in \bLtwo$ and $\bfpsi \in \mathbf{H}^2$ such that $\q = -\tA\bfpsi $, we have:
    \begin{align}
        ( \q, \Delta \bfpsi) =  \vertiii{\q}_0^2  \label{eq:ch1(u,v)=|u|2}\\
      |\Delta \bfpsi|^2 = |\Delta \bfpsi|^2\leq \vertiii{\q}_{0}^2 \leq c_0 |\Delta \bfpsi|^2 . \label{eq:ch12normpoincare}
    \end{align}
\end{lemma}

\autoref{tab:norms} contains a summary of the spaces and relative norms used throughout this work. 
\begin{table}[ht]
    \centering
    \begin{tabular}{ll}
    Space & Norm \\[1ex] 
     $\bH^{k}$ & $\|\bfpsi\|_k^2 = h_1\|\psi_1\|_k^2 + h_2\| \psi_2\|_k^2$ \\ 
     $\bLtwo$ = $\bH^0$ & $|\bfpsi|^2= h_1|\psi_1|^2 + h_2|\psi_2|^2$ \\ 
     $\bH^1 $& $\|\bfpsi\|^2 = h_1\|\psi_1\|^2 + h_2\| \psi_2\| ^2$ \\[1ex]
      $\bLtwo$
      & $\vertiii{\q}_0^2 = | \Delta\bfpsi|_2^2 + p \| \psi_1 - \psi_2\|^2 $\\ 
      $\Hmo$
      & $\vertiii{\q}_{-1}^2 = \| \bfpsi \|^2 + p | \psi_1 - \psi_2|^2$
      
    \end{tabular} 
    \caption{\label{tab:norms} Notations for the two--layer quasi--geostrophic model. 
            Rows~1--3 will be used for the streamfunctions $\bfpsi$. 
            Rows~4,5 will be mainly used for the potential vorticities $\q$, that is $\q\in \bLtwo$, $\Hmo$ respectively and $\bfpsi$ is such that $\q = -\tA \psi$.}
\end{table}

     Finally, standard bounds on  the Jacobian (see for example \cite[Lemma~3.1]{Chueshov01_Proba}), yield the following bound for the bilinearity $B$:
    \begin{lemma}[{\cite[Lemma~1.3.4]{thesis}}]
    \label{lemma:ch1propertiesB}
    Let $B$ be the bilinear operator defined in \eqref{eq: def B(U,V)}, then for  $\bfpsi,\bfxi, \bfphi\in \bH^2$ 
    \begin{align}
        (B(\bfpsi,\bfxi), \bfphi) & = - (B(\bfphi,\bfxi), \bfpsi), \label{eq: (B(U,V), W)= - (B(W,V), U) } \\
       (B(\bfpsi,\bfxi), \bfpsi) & = 0. \label{eq: (B(U,V), U) = 0} 
    \end{align}
    Moreover, for $\bfpsi, \bfxi \in \bH^2$, there exists positive constant $k_0$ such that
    \begin{equation}
        |(B(\bfpsi,\bfpsi),\bfxi)|\leq k_0 \| \bfpsi\| |\Delta \bfpsi|  |\Delta \bfxi| .
        \label{eq:bound(B(u,u),v)}
    \end{equation}
    \end{lemma}
The deterministic version of \eqref{eq:QG_stoc_vec} has been shown to be well--posed in \cite{Bernier94}. For the stochastic model \eqref{eq:QG_stoc_vec}, for $f\in H^{-2}$, $\q_0\in \Hmo$ and $T>0$ there exists a pathwise unique solution $\q(t, \omega;\q_0, a)$ in $ C([0,T]; \Hmo) \cap L^2(0,T; \bLtwo) $ for almost all $\omega$ and the associated Markov semigroup $\cP_t^a$ is Feller as $\q$ is a continuous function of the initial condition $\q_0$ as a function in $\Hmo$. For a complete proof of these result we refer to \cite[Section~2]{thesis} and references therein.
    \subsection{Spectral gap}\label{subsec:QGspectralgap}
    In \cite{carigi2022exponential} it is shown that \eqref{eq:QG_stoc_vec} exhibits a spectral gap as it satisfies \autoref{thm:expconvergence}, 
    provided the parameter $r$ is large enough.
    More specifically, it is demonstrated for a fixed $a$ \nameref{asuA} holds with Lyapunov function $V(\q) = \vertiii{\q}_{-1}^2$. The chosen controlled equation is
    \begin{equation*}
        d\tq + \left( B(\tpsi, \tpsi) + \beta \partial_x \tpsi \right)\, dt = \nu \Delta^2 \tpsi + \binom{f(a) + r\Pi_n(\psi_1 - \tilde{\psi}_1)}{- r\Delta \tilde{\psi}_2}
    \end{equation*}
    where $\Pi_n$ is the projection onto the first eigenfunctions of $-\Delta$ for an appropriate $n$. Similarly to the Navier--Stokes equation, the dependence on the forcing $f$ is explicit in the estimates needed for \nameref{asuA}, hence the calculations in \cite{carigi2022exponential} are easily extendable to estimates with parameters independent from $a$ as desired in this context. Table~\ref{tab:parametertranslateQG} shows the parameters mentioned in \nameref{asuA} and how they relate to the parameters appearing in the stochastic 2LQG~equations.
    \begin{table}
    \begin{center}
    \begin{tabular}{ll}
    \nameref{asuA} & 2LQG \\ [1ex]
    $\kappa_0$ &  $r$ \\ 
    $\kappa_1$ &  $k_0^2/\nu$ \\ 
    $\kappa_2$ &  $\nu - 2\gamma \Tr Q/\lambda_1^2$\\ 
    $\kappa_\varepsilon$ &  $T_Q + h_1 \sup_{|a-a_0|< \varepsilon}\|f(a)\|_{-2}^2/\nu$\\ 
    $\gamma_\varepsilon$ &  $\nu\lambda_1/c_0 $\\ 
    $K_\varepsilon$ &  $T_Q + h_1 \sup_{|a-a_0|< \varepsilon}\|f(a)\|_{-2}^2/\nu$
    \end{tabular} 
    \end{center}
    \caption{\label{tab:parametertranslateQG}
    Parameters from \nameref{asuA} and their values in terms of the parameters appearing in the stochastic 2LQG~equations.
        Here $k_0$ appears in the estimates of the trilinear form in \autoref{lemma:ch1propertiesB}, $\nu$ the viscosity, $\gamma$ is chosen so that $\kappa_2$ is positive, $\lambda_1$ is the smallest eigenvalue of $-\Delta$, $Q$ the covariance operator of the noise, $T_Q = \Tr (Q^{1/2})^* \tA^{-1} Q^{1/2}$, $f(a) $ is the deterministic external forcing and $h_1$ is the height of the top layer.
    }
    \end{table}
    Finally, the required lower bound on the parameter $r$ is  
    \begin{equation}
    \label{eq:condition_r}
         r   > \tfrac{2k_B}{ \nu } \left(\tfrac{h_1}{\nu}\|f(a)\|_{-2}^2 + T_Q\right) .
    \end{equation}

    \subsection{Linear and fractional response}
    \label{subsec:QGHolder}
     We start with the weak differentiability or linear response for the two--layer quasi-geostrophic model.
     This was the main motivation for developing the general methodology presented above which now affords a very concise proof.
    \begin{theorem}
    \label{thm:QGresponse}
        Set $\Ieps = (a_0 - \varepsilon, a_0 + \varepsilon)\subset \R$ and consider the two--layer quasi--geostrophic equation \eqref{eq:QG_stoc_vec} with parameters satisfying \eqref{eq:condition_r}. Suppose $a\mapsto f(a)$ is continuously differentiable as a map from $\Ieps$ into $\range Q$, with $|D_a Q^{-1/2}f |$ uniformly bounded in $a$, and let $\mu_a$ be the associated unique invariant measure. Then the map $a \mapsto\langle \varphi, \mu_a\rangle $ is differentiable at $a = a_0$ for every $\varphi\in \Cd$ with 
         \begin{equation*}
            \left. \dv{}{a} \langle \varphi, \mu_a \rangle\right|_{a = a_0} = \langle D_a \cP_t^{a}|_{a = a_0} (1 - \cP_t^{a_0})^{-1}(\varphi - \langle \varphi, \mu_{a_0}\rangle)  , \mu_{a_0}\rangle. 
        \end{equation*}
    \end{theorem}
    \begin{proof}
      Since \eqref{eq:condition_r} is satisfied then, as discussed in \autoref{subsec:QGspectralgap}, \nameref{asuA} holds. Then the thesis follows from \autoref{thm:responseU}.
    \end{proof}
   For this model the requirement for the forcing to be in the range of the noise implies that the forcing has to act on the same layer where the noise does. This is a natural assumption in some applications for example if the random term accounts for changes in the intensity of the average wind forcing on the upper ocean. However, the presented methodology cannot deal with the response to changing forcings in the bottom layer.
    Yet for such forcings that are not necessarily in the range of the noise, we can nevertheless show fractional response.
    \begin{theorem}
    \label{thm:QGHolder}
        Consider \eqref{eq:QG_stoc_vec} with parameters satisfying \eqref{eq:condition_r} and parameter $a$ in the interval $\Ieps = (a_0 - \varepsilon, a_0 + \varepsilon)\subset \R$. Suppose $f(a)$ is $\beta$--H\"{o}lder continuous as a function from $\Ieps$ into $H^{-2}$, namely there exists $C_f= C_f(a_0, \varepsilon)$ such that 
        \begin{equation*}
            \|f(a_1) - f(a_2) \|_{-2}\leq C_f |a_1 - a_2|^\beta \fa a_1, a_2 \in \Ieps.
        \end{equation*}
        Then for all $\alpha \in (0, \alpha_0)$, with $\alpha_0$ as in \eqref{eq:QGupsilonalpha0}, there exists $c=c(\varepsilon)$ such that 
        \begin{equation*}
            |\langle \varphi, \mu_{a_1} - \mu_{a_2}\rangle |\leq c\|\varphi\|_{\td} |a_1 - a_2|^{\alpha\beta} \fa \varphi\in \Cd
        \end{equation*}
        for all $a_1, a_2\in \Ieps$, namely the map $a \mapsto  \mu_a $ is locally $\alpha\beta$-H\"{o}lder continuous.
    \end{theorem}
    \begin{proof}
    We want to ensure that \nameref{asuH} holds to apply \autoref{thm:SPDEHolder}. By the results in the literature (\cite{carigi2022exponential}) we know \nameref{asuA} holds, as discussed in \autoref{subsec:QGspectralgap}, and hence \ref{H1} holds. 
    \paragraph{Proof of \ref{H3}} As seen for the Navier--Stokes equation it is sufficient to check that the conditions of \autoref{lemma: Lemma5.1HMatt08} are met.
        Taking the $\bLtwo$ product of \eqref{eq:QG_stoc_vec} with $\eta \bfpsi = - \eta \tA^{-1}\q$, using \eqref{eq:ch1(u,v)=|u|*} and \eqref{eq: (B(U,V), U) = 0} we have 
        \begin{equation*}
            d(\eta \vertiii{\q}_{-1}^2) =  - 2 \eta \left(  \nu |\Delta \bfpsi|^2 + h_1 \langle f(a), \psi_1\rangle + rh_2\| \psi_2\|^2 \right) \,dt \\  +\eta T_Q \, dt -  2\eta h_1\left( \psi_1, dW(t)\right) .
        \end{equation*}
       where $T_Q = \Tr (Q^{1/2})^* \tA^{-1} Q^{1/2}$. 
       We check the conditions of \autoref{lemma: Lemma5.1HMatt08} with $M(t) := \eta \vertiii{\q(t)}^2_{-1}$, $Z(t) := \eta \lambda_1^{-1} \vertiii{\q(t)}_0^2$.
    \begin{enumerate}[label=(\roman*)]
        \item Set the function $F(t)$ to be 
        \begin{equation*}
             F(t) := - 2 \eta \left(h_1 \langle f(a), \psi_1\rangle + rh_2\| \psi_2\|^2 + \nu |\Delta \bfpsi|^2  \right) + \eta T_Q. 
        \end{equation*}
       By Cauchy-Schwartz, Young and Poincar\'{e} inequalities and dropping the term $rh_2\|\phi_2\|$ we get
        \begin{align*}
           F(t) & \leq  \frac{\eta h_1}{\nu}\| f(a) \|^2_{-2} + \eta \nu h_1|\Delta \psi_1| ^2   - 2 \eta \nu |\Delta \bfpsi|^2 + \eta T_Q\\
           & \leq \eta \left(\frac{ h_1}{\nu }\| f(a) \|_{-2}^2 +  T_Q \right) - \eta \nu |\Delta \bfpsi|^2 .
        \end{align*}
        Then estimating $|\Delta \bfpsi|^2$ by \eqref{eq:ch12normpoincare}, i.e.~$\vertiii{\q}_0^2 \leq c_0 |\Delta \bfpsi|^2$, 
        \begin{equation*}
            F(t) \leq b_1 - b_2 \eta \lambda_1^{-1}\vertiii{\q(t)}_0^2
        \end{equation*}
        with 
        \[ 
        b_1 = \eta \left(\frac{ h_1}{\nu }\| f(a) \|_{-2}^2 +  T_Q \right) \quand b_2 =   \frac{\nu\lambda_1}{c_0}. 
        \]
        \item  By Poincar\'{e} inequality $M(t) = \eta \vertiii{\q(t)}^2_{-1} \leq \eta \lambda_1^{-1} \vertiii{\q(t)}_0^2 = Z(t)$.
        \item It is easy to see that there is a standard real-valued Wiener process $B(t, \omega)$ such that 
        \begin{equation*}
            -2\eta h_1(\psi_1(t), dW(t)) = - 2 \eta h_1 \| (\psi_1(t), \cdot )\|_{L_2^0}^2 \, dB(t) =: G(t)\, dB. 
        \end{equation*}
        We then have to ensure that there exists $b_3 \in (0, b_2)$ such that $G^2 \leq b_3 Z$ almost surely i.e.~
       \begin{equation}
        \label{eq:ch5proofqgiii}
             4\eta^2 h_1^2\| (\psi_1, \, \cdot ) \|_{L_2^0}^2 \leq b_3 \eta \lambda_1^{-1} \vertiii{\q(t)}_0^2.
        \end{equation}
        By definition of $L_2^0$ we have 
        \begin{equation*}
             \|(\psi_1, \cdot) \|_{L_2^0}^2 =\sum_{k\in \N} |(\psi_1, Q^{1/2}e_k)|^2.
        \end{equation*}
       By Cauchy-Schwartz inequality and Poincar\'{e} inequality it follows
        \begin{equation*}
            \|(\psi_1, \cdot) \|_{L_2^0}^2 \leq |\psi_1|^2 \sum_{k\in \N}  |Q^{1/2}e_k|^2 = \Tr Q  |\psi_1|^2 \leq  \frac{ \Tr Q}{\lambda_1^{2}h_1} |\Delta \bfpsi|^2
        \end{equation*}
         and consequently by \eqref{eq:ch12normpoincare}, \eqref{eq:ch5proofqgiii} holds setting $b_3 = 4 \eta h_1 \lambda_1^{-1} \Tr Q$. 
    \end{enumerate}
    As we require $b_2> b_3$ i.e.~
    \begin{equation*}
       \frac{\nu\lambda_1}{c_0} > 4 \eta h_1 \lambda_1^{-1} \Tr Q 
    \end{equation*}
    we get that for all 
    \begin{equation*}
        0< \eta < \frac{\nu\lambda_1^2}{4 c_0 h_1 \Tr Q}=:\eta_1,
    \end{equation*}
    the hypothesis of \autoref{lemma: Lemma5.1HMatt08} hold, giving 
        \begin{equation*}
              \E \exp(\eta \vertiii{\q(t)}_{-1}^2) \leq c(a) \exp(\eta \vertiii{\q(0)}_{-1}^2e ^{-b_2 t/2})
        \end{equation*}
        for all $\eta < \eta_1$ and all $t > 0$ with 
        \[ 
            c(a) = \frac{b_2 \exp(\frac{2b_1}{b_2})}{b_2 - b_3} = \frac{\frac{\nu\lambda_1}{c_0} \exp(\frac{2c_0 \eta}{\nu\lambda_1}\left(\frac{ h_1}{\nu \lambda_1}\|f(a) \|^2_{-1} + T_Q \right))}{\frac{\nu\lambda_1}{c_0} - 4 \eta h_1 \lambda_1^{-1} \Tr Q}.
        \]
        which stays uniformly bounded for all $a\in \Ieps$. We have only to show that $\eta = \alpha_0 \upsilon < \eta_1$.
        By definition of $\alpha_0$ we have that
        \begin{equation*}
           \alpha_0 \upsilon = \upsilon \left( \frac{1}{2} \wedge \frac{2\gamma}{2\gamma + \upsilon}\right)< 2\gamma.
        \end{equation*}
        Recall that $\gamma$ is an arbitrary parameter introduced so that $\kappa_2$ stays positive, i.e.
        \begin{equation*}
            \gamma < \frac{\nu \lambda_1^2}{2\Tr Q}.
        \end{equation*}
        Then picking 
        \begin{equation*}
            2\gamma< \frac{\nu \lambda_1^2}{2\Tr Q}\left( 1 \wedge (2c_0h_1)^{-1}\right)\leq \eta_1,
        \end{equation*}
        it is ensured that $\alpha_0\upsilon< \eta_1$.
        
    \paragraph{Proof of \ref{H2}} Consider $\q(t)$ and $\tq(t)$, unique solutions of \eqref{eq:QG_stoc_vec} respectively with parameter $a_1$ and $a_2$, and same realization of the noise. Then the difference $\bfxi :=\q- \tq$, with corresponding streamfunction $\bfphi := \bfpsi - \tilde{\bfpsi}$, satisfies the following equation 
        \begin{align*}
            \begin{split}
            &\dv{\bfxi}{t} + B(\bfphi, \bfpsi) + B(\tilde{\bfpsi}, \bfphi) + \beta \partial_1 \bfphi= \nu \Delta^2\bfphi + \left( \begin{array}{c} f(a_1) - f(a_2) \\ -r\Delta \phi_2 \end{array} \right)\\
            & \bfxi = \left( \Delta + M \right)\bfphi,\\
            &\bfxi(0) = 0 
        \end{split}
        \end{align*}
        To bound $\vertiii{\bfxi}_{-1}$ we take the $\bLtwo$ product with $\bfphi$: by the properties of the nonlinearity and Green theorem we get  
        \begin{equation*}
            \tfrac{1}{2}\ddt\vertiii{\bfxi}_{-1}^2 + \nu |\Delta\bfphi|^2 + rh_2\|\phi_2\|^2 = (B(\tpsi, \bfphi), \bfphi) - h_1 \langle f(a_1) - f(a_2) , \phi_1 \rangle  .
        \end{equation*}
        By the bound on the bilinearity \eqref{eq:bound(B(u,u),v)} and Young's inequality we have 
        \begin{equation*}
             |(B(\tpsi, \bfphi), \bfphi)|\leq \tfrac{k_0^2}{2\nu}|\Delta \tpsi|^2\|\bfphi\|^2 + \tfrac{\nu}{2}|\Delta \bfphi|^2,
        \end{equation*}
        and by Young's inequality and the H\"{o}lder continuity of $f$
        \begin{align*}
            h_1 |\langle f(a_1) - f(a_2), \phi_1 \rangle | 
            &\leq \tfrac{h_1}{\nu}\|f(a_1) - f(a_2)\|^2_{-2} + \tfrac{\nu  h_1}{4}| \Delta \phi_1 |^2 \\
             &\leq \tfrac{h_1C_f^2}{\nu}|a_1 - a_2|^{2 \beta}+ \tfrac{\nu  h_1}{4}| \Delta \phi_1 |^2
        \end{align*}
        so that 
        \begin{equation*}
            \tfrac{1}{2}\ddt\vertiii{\bfxi}_{-1}^2 +  \tfrac{\nu}{2} |\Delta \bfphi|^2 + rh_2\|\phi_2\|^2\leq  \tfrac{k_0^2}{2\nu}|\Delta \tpsi|^2 \|\bfphi\|^2  +  \tfrac{h_1C_f^2}{\nu} |a_1 - a_2|^{2 \beta}+ \tfrac{\nu  h_1}{4}| \Delta \phi_1 |^2.
        \end{equation*}
        Finally, rearranging and using \eqref{eq:ch1*normpoincare} we have dropping the second and third term on the left hand side
         \begin{equation*}
            \ddt\vertiii{\bfxi}_{-1}^2 \leq  \tfrac{k_0^2}{\nu}|\Delta \tpsi|^2  \vertiii{\bfxi}_{-1}^2  + \tfrac{h_1C_f^2}{\nu}|a_1 - a_2|^{2\beta},
        \end{equation*}
        and by Gronwall's lemma 
        \begin{equation*}
            \vertiii{\bfxi(t)}_{-1}^2 \leq   \tfrac{h_1C_f^2}{\nu}|a_1 - a_2|^{2\beta} \int_0^t \exp\left( \tfrac{k_0^2}{\nu}\int_s^t |\Delta \tpsi|^2 \, d\tau\right) \, ds .
        \end{equation*}
         Therefore \ref{H2} is satisfied with $C= C(t) = t h_1  C_f^2/\nu  $. 
    \end{proof}



\paragraph{} Dipartimento di Ingegneria e Scienze dell’Informazione e Matematica, Universit\`{a} degli Studi dell’Aquila, 67100 L’Aquila, Italy\\
Centre for the Mathematics of Planet Earth, University of Reading, Reading, RG6 6AX, UK.\\
\textit{Email:} giulia.carigi@univaq.it

\paragraph{} Dipartimento di Ingegneria e Scienze dell’Informazione e Matematica, Universit\`{a} degli Studi dell’Aquila, 67100 L’Aquila, Italy\\
Centre for the Mathematics of Planet Earth, University of Reading, Reading, RG6 6AX, UK.\\
\textit{Email:} tobias.kuna@univaq.it

\paragraph{} Department of Mathematics and Statistics, Department of Meteorology,\\ and Centre for the Mathematics of Planet Earth,\\ University of Reading, Reading, RG6 6AX, UK.\\
\textit{Email:} j.broecker@reading.ac.uk

\end{document}